\documentclass[a4paper,USenglish,cleveref, autoref, thm-restate]{lipics-v2021}
\usepackage{lineno} 
\nolinenumbers
\newcommand{\Nat}{I\!\!N}

\def\idtt#1{\ensuremath{\mathtt{#1}}}

\def\op#1{\idtt{#1}}

\usepackage{tikz}
\usetikzlibrary{arrows.meta,decorations,decorations.pathreplacing,calc,bending,positioning, 
	chains,patterns,scopes}

\tikzstyle{box}=[draw, outer sep=0pt, inner sep=0pt ,on chain, font=\small]
\tikzstyle{dottedline}=[ultra thick, loosely dotted,shorten >=1mm, shorten <=1mm]

\usetikzlibrary{chains, fit, positioning}

\newcommand{\field}[4][white]{%
	\node [box,fill=#1] (#2) {#3};
	\node [below=0 of #2] (#2i) {#4};%
}

\newcommand{\connectBelow}[3][0.6]{%
	\draw[-] (#2) |- ++(0,+#1) -| (#3);%
}

\usepackage{tikz}
\usepackage{tikz-qtree}
\usepackage{listings}
\usetikzlibrary{arrows,automata,arrows.meta,decorations,
	decorations.pathreplacing,calc,bending,positioning,chains,patterns,matrix}
\usetikzlibrary{decorations.markings}
\usetikzlibrary{decorations.shapes}
\usetikzlibrary{shapes.geometric}
\tikzset{paint/.style={ draw=#1!50!black, fill=#1!50 },
	decorate with/.style=
	{decorate,decoration={shape backgrounds,shape=#1,shape size=2mm}}}

\definecolor{slblue}{HTML}{CCFFFF}
\definecolor{tree1}{HTML}{ffb74d}
\definecolor{tree2}{HTML}{aed581}
\definecolor{tree3}{HTML}{e57373}
\definecolor{tree4}{HTML}{7986cb}

\usepackage{amsfonts}

\usepackage{enumitem}

\usepackage{subcaption}
\usepackage{graphicx}
\usepackage{amsmath}

\usepackage{lineno}
\usepackage{todonotes}

\newcommand*{\bb}{N}
\usepackage{hyperref}

\hideLIPIcs


\title{Sorting and Ranking of Self-Delimiting Numbers with Applications to  
Outerplanar Graph Isomorphism}

\author{Frank Kammer}{THM, University of Applied Sciences Mittelhessen, Giessen, Germany}{Frank.Kammer@mni.thm.de}{https://orcid.org/0000-0002-2662-3471}{}

\author{Johannes Meintrup}{THM, University of Applied Sciences Mittelhessen, Giessen, Germany}{Johannes.Meintrup@mni.thm.de}{https://orcid.org/0000-0003-4001-1153}{}

\author{Andrej Sajenko}{THM, University of Applied Sciences Mittelhessen, Giessen, Germany}{Andrej.Sajenko@mni.thm.de}{https://orcid.org/0000-0001-5946-8087}{}


\funding{J. Meintrup and A. Sajenko are funded by the Deutsche Forschungsgemeinschaft (DFG, German Research Foundation) – 379157101.}

\relatedversion{This is the full version of the paper~\cite{KamMS23}, which includes  proofs and result to forest and outerplanar graph isomorphism.}

\authorrunning{F. Kammer, J. Meintrup and A. Sajenko}
\titlerunning{Sorting and Ranking of Self-Delimiting Numbers and Outerplanar Graph Isomorphism}

\begin{document}





\maketitle

\begin{abstract}
    Assume that an $\bb$-bit sequence $S$ of $k$ numbers encoded as Elias gamma
    codes is given as input. We present space-efficient algorithms for sorting,
    dense ranking and competitive ranking on $S$ in the word RAM model with
    word size $\Omega(\log \bb)$ bits. Our algorithms run in $O(k +
    \frac{\bb}{\log \bb})$ time and use $O(\bb)$ bits. The sorting algorithm
    returns the given numbers in sorted order, stored within a bit-vector of
    $\bb$ bits, whereas our ranking algorithms construct data structures that
    allow us subsequently to return the dense/competitive rank of each number
    $x$ in $S$ in constant time. For numbers $x \in \Nat$ with $x > \bb$ we
    require the position $p_x$ of $x$ as the input for our
    dense-/competitive-rank data structure. As an application of our algorithms
    above we give an algorithm for tree isomorphism, which runs in $O(n)$ time
    and uses $O(n)$ bits on $n$-node trees. Finally, we generalize
    our result for tree isomorphism to forests and outerplanar graphs,
    while maintaining a space-usage of $O(n)$ bits. The previous best
    linear-time algorithms for trees, forests and outerplanar graph
isomorphism all use $\Theta(n \log n)$ bits.
\end{abstract}


\section{Introduction}\label{sec:intro}

Due to the rapid growth of input data sizes in recent years, 
algorithmic solutions that optimize both run time and memory usage 
have become increasingly important. For this reason, we focus on {\em space-efficient}
algorithms, i.e., algorithms that run (almost) as fast as standard solutions for
the problem under consideration, while decreasing the required space by more than 
an asymptotic constant.

Graphs are often used to encode structural information in many fields, e.g., in
chemistry or electronic design automation.
Isomorphism testing on these graphs can then be used for
identifying a chemical compound within a large chemical database~\cite{IrnB04}.
Moreover, it is the basis
for a verification whether large electric circuits represented by a circuit
schematic and an integrated circuit layout correspond~\cite{BaiC75}.

{\bfseries Model of Computation.} Our model of computation is the word RAM,
where we assume to have the standard operations to read, write as well as
arithmetic operations (addition, subtraction, multiplication, modulo, bit shift,
\texttt{AND} and \texttt{OR}) take constant time on words of size $w = \Omega(\log \bb)$ bits
where $\bb$ is the input size in bits. (In our paper $\log$ is the binary
logarithm~$\log_2$.) The model has three types of memory. A read-only {\em input
memory}, a read/write {\em working memory} and a write-only {\em output memory}.
The space bounds stated for
our results are for the working memory.

{\bfseries Sorting.} Sorting is one of the most essential concepts for the
design of efficient algorithms and data
structures~\cite{AsanoEK13,BorC82,Chan10,Fre87,Hagerup98,IsaacS56} {for over} 60
years. For an overview of several sorting algorithms see~\cite{Knuth98a}.
Sorting problems are usually classified into different
categories. In {\em comparison sorting} two elements of an input sequence must
be compared against each other in order to decide which one precedes the other.
Pagter and Rauhe~\cite{PagR98} gave a comparison-based algorithm that runs on
input sequences of~$k$ elements in $O({{k^2}/s})$ time by using $O(s)$ bits for
every given $s$ with $\log k \le s \le{k / {\log k}}$. Beame~\cite{Bea91}
presented a matching lower bound in the so-called \emph{branching-program
model}.
Let $[0, x] = \{0, \ldots, x\}$ and $[0, x) = \{0, \ldots, x - 1\}$ for any
natural number~$x$. {\em Integer sorting} asks to sort a
sequence of $k$ integers, each in the range $[0, m)$. It is
known that, for $m = k^{O(1)}$, integer sorting can be done in linear time:
consider the numbers as $k$-ary numbers, sort the digits of the numbers in
rounds (radix sort) and count the occurrences of a digit by exploiting indirect
addressing (counting sort).  
Pagh and Pagter showed optimal time-space trade-offs for integer
sorting~\cite{PagP02}. 
Moreover, Han showed that {\em real sorting} (the given sequence consists of
real numbers) can be converted in $O(k\sqrt{\log k})$ time into integers and
then can be sorted in $O(k\sqrt{\log\log k})$ time~\cite{Han18}.

These algorithms above all assume that the numbers of the input are represented
with the same amount of bits. We consider a special case of integer sorting that
appears in the field of space-efficient algorithms where numbers are often
represented as so-called {\em self-delimiting numbers} to lower their total
memory usage. A self-delimiting number can be represented in several ways. {We
use the following straightforward representation, also known as \textit{Elias
gamma code}. To encode an integer $x>0$ write $\ell = \lfloor \log x \rfloor$
zero bits, followed by the binary representation of $x$ (without leading zeros).
When needed, the encoding can be extended to allow encoding of integers $x \geq
0$ by prepending each encoded $x>0$ with a single bit set to $1$, and encoding
$0$ with a single bit set to $0$. Thus, the self-delimiting number of an integer
$x$ can be stored within $2 \lceil \log x \rceil + 1$ bits. As an example, the
self-delimiting numbers of $0,1,2,3,4$ are $0, 11, 1010, 1011, 100100$,
respectively. Throughout this paper, we assume all self-delimiting numbers are
given as Elias gamma codes.} Our results can be adapted to other types of
self-delimiting numbers, for example, Elias delta and Elias omega codes. The
property we require is the following: let $x_1 < x_2$ be two integers encoded as
self-delimiting numbers $e_1, e_2$, respectively. Then it holds that $e_1$ uses
at most as many bits as $e_2$, and if they use the same number of bits, then
$e_1$ is lexicographically smaller than $e_2$.

Assume that $k$ self-delimiting numbers in the range $[0, m)$ with
$m=O(2^N)$ are stored in an $\bb$-bit sequence. If the memory is
unbounded, then we can simply transform the numbers into integers, sort them,
and transform the sorted numbers back into self-delimiting numbers. However,
this approach uses $\Omega(k \log m)$ bits. For $k \approx \bb \approx m$, this
is too large to be considered space-efficient.
We present a sorting algorithm for self-delimiting numbers that runs in
$O(k+\frac{\bb}{\log \bb})$ time and uses $O(\bb)$ bits.

{\bfseries Ranking.} {Let $S$ be a sequence of numbers. The \textit{competitive
rank} of each $x \in S$ is the number of elements in $S$ that are smaller
than~$x$. The \textit{dense rank} of each $x \in S$ is the number of distinct
elements in $S$ that are smaller than $x$, i.e., competitive rank counts
duplicate elements and dense rank does not.}

{Raman et al.~\cite{RamRR02}  presented a data structure for a given a set
$S\subseteq [0,m)$ of $k$ numbers that uses $\Theta(\log \binom{m}{k})=\Omega(k
\log (m/k))$ bits}
to answer dense rank (and other operations) in constant time. In some sense,
this space bound is ``optimal'' due to the entropy bound, {assuming} we treat
all numbers in the same way. However, the representation of the self-delimiting
numbers differs in their size. E.g., we have a bit vector of $N$ bits
storing self-delimiting numbers such that the vector consists of $\Theta(N)$
numbers where one number is $2^{\Theta(N)}$ and all other numbers are $1$. Then,
the space bound above is $\Omega(N \log (2^{\Theta(N)}/N))=\Omega(N^2)$.

We present an algorithm to compute the {dense/competitive rank} on a sequence
$S$ of length $\bb$ consisting of $k$ self-delimiting numbers in $O(k +
\frac{\bb}{\log \bb})$ time using $O(\bb)$ bits and {subsequently} answer
dense/competitive rank queries of a number $x \in S$ in constant time. For
numbers $x$ of size $>N$ we require the position $p_x$ of $x$ in $S$ as the
input to the respective query.

{\bfseries
Tree, Forest and Outerplanar Isomorphism.}
In the last decade, several space-efficient graph algorithms have been
published. Depth-first search and breadth-first search are the first problems
that were considered~\cite{AsaIKKOOSTU14,BanerjeeNCSRVSR18,ElmHK15,Hag18}. 
Further papers with
focus on space-efficient algorithms discuss graph
interfaces~\cite{BarAHM12,HagKL19,KamM23}, connectivity
pro\-blems~\cite{ChoGS18,Hag18}, matching~\cite{DattaBK16},
separators and treewidth~\cite{IzumiO20,KammerM22,KammerMS22}
and temporal graphs~\cite{HeegerHKNRS21}. Several of these
results are implemented in an open source GitHub project~\cite{KamS18g}. 
We continue this research and present a space-efficient isomorphism algorithm for
trees, based on an algorithm described in the textbook of Aho,
Hopcroft and Ullman~\cite{AhoHU74}, which uses $\Omega(n \log n)$ bits. We
improve the space-bound to $O(n)$ bits while maintaining the linear running
time.
Lindell~\cite{Lindell92} showed that tree isomorphism is in $\textsc{L}$;
however, on $n$-{node} instances the algorithm presented has a running time of~$\Omega(n^2)$:
in one subcase of his
algorithm, two subtrees with $\Theta(n)$ vertices of each of the given trees have to
be cross compared recursively with each other. A runtime of~$\Omega(n^2)$
follows.
We present an $O(n)$-time and $O(n)$-bit tree-isomorphism algorithm that decides if 
two given unrooted unlabeled $n$-node trees
are isomorphic.
Colbourn and Booth~\cite{ColbournBK81} showed a linear time algorithm for
outerplanar isomorphism that uses $\Omega(n \log n )$ bits, and a simpler
algorithm of Beyer and Mitchell~\cite{MitBJ79} for the special case of maximal
outerplanar graphs exists with the same runtime and space usage. We show how a
space-efficient version of both these results can be implemented with the use of our
novel tree-isomorphism algorithm.

{\bfseries Outline.} We continue our paper with some general preliminaries,
followed by our results on sorting and
dense/competitive ranking in Section~\ref{sec:sort-rank}. Afterwards we
introduce definitions and notations in Section~\ref{sec:prelim} as a preparation
for our result on space-efficient tree isomorphism
(Section~\ref{sec:treeiso}),
which we generalize to forests in Section~\ref{sec:forest-iso}. 
Finally, we present our isomorphism algorithm for (maximal) outerplanar graphs is Section~\ref{sec:outerplanar}.

\section{Preliminaries}\label{sec:preone}

All input graphs $G=(V, E)$ 
are assumed to be given via adjacency
arrays, i.e., for each
vertex $u \in V$ we have access to an array $\op{adj}[u]$ that allows us to get a $k$th neighbor $v$ of $u$
via $\op{adj}[u][k]$ in constant time.
A common implementation of this interface is a $2$-dimensional array,
where the labels of the vertices are implicitly defined by the order
of the adjacency arrays, e.g., $\op{adj}[u]$ refers to an array containing
all neighbors of the vertex labeled $u$.

We require some commonly used data structures in the field
of space-efficient algorithms. Firstly, a simple
(dense) rank-select data structure that we use throughout as
a utility. We discuss ranking in more detail in the 
next section. 
\begin{lemma}{(rank-select~\cite{BauH17})}\label{lem:rs} Given access to a sequence $B = (b_1,
	\ldots, b_n) = \{0, 1\}^{n}$ $(n \in \Nat)$ of $n$ bits there is an $o(n)$-bit
	data structure that, after an initialization of $O(n / w)$ time, supports two
	constant-time operations:
	\begin{itemize}
		\item $\op{rank}_B(j) = \sum_{i=1}^{j}b_i$ ($j \in [1, n]$) that returns
	the number of $1$s in $(b_1, \ldots, b_j)$ in $O(1)$ time, and
		\item $\op{select}_B(k) = \min \{j \in [1, n] : \op{rank}_B(j) = k\}$
	that returns the position of the $k$th one in $B$.
	\end{itemize}
\end{lemma}

We use a folklore technique called \emph{static space allocation}, which allows
us to store $\ell$ items of different size in a compact manner, e.g., a sequence
of self-delimiting numbers. The following formal description is adapted
from~\cite{KamKL19}. Let $S$ be a sequence of $\ell$ items $B_k$ such that each
item $B_k$ occupies $d_k$ bits for $k \in \{1, \ldots, \ell\}$ and let
$N=\sum_{k \in \{1, \ldots, \ell\}} d_k$. Using a simple bitstring $B$ of length
$N$ with a bit set to $1$ exactly at all indices $k+\sum_{j=1}^{k-1} d_k$ for
all values $k$, following by constructing a rank-select data structure for $B$,
one can easily gain random-access to each element $B_k \in S$ as obtaining the
index of the first bit that belongs to $B_k$ as $\op{select}_B(k)-k$.
This requires $|S+\ell|+o(|S+\ell|)$ bits in addition to the bits required to store $S$.

A final data structure we require is a \emph{choice dictionary}.

\begin{lemma}{(choice dictionary~\cite{Hag19,KamS18c})}\label{lem:cd} 
	There is a data structure $C$ initialized for universe $U=\{1, \ldots, \ell\}$
	that occupies $\ell+o(\ell)$ bits and is initialized in $O(\ell)$ time.
	The data structure allows the following constant time (per element) operations:
	\begin{itemize}
		\item $\op{insert}(i)$: Insert $x\in U$ into $C$
		\item $\op{remove}(i)$: Remove $x\in U$ from $C$
		\item $\op{contains}(i)$: Return $\op{true}$ exactly if $x \in C$.
		\item $\op{choice}(i)$: Return an arbitrary element $x$ from $C$
		\item $\op{iterate}$: Iterate over all elements $x \in C$.
	\end{itemize}
\end{lemma}

\section{Sorting and Ranking}\label{sec:sort-rank} In this section we consider
sorting and ranking of~$k$ self-delimiting numbers, stored within an $\bb$-bit
sequence~$S$.
We make use of {\em{lookup tables}}, which are precomputed tables storing the
answer for every possible state of a finite universe, typically of small size.
We use such tables to quickly solve sub-problems that occur in our algorithms.
For the rest of this section we assume that a parameter $\tau$ with $\log N \le
\tau \le w$ (where $w$ represents the word size) is given
to our algorithms, which we use as a parameter to construct lookup tables for
binary sequences of length $\lceil\tau/2\rceil$. To give an intuitive example of
such lookup tables, the universe might consist of all integers of size at most
$2^{\lceil\tau/2\rceil}$ and answer queries if the given number is prime. Such a
table has $2^{\lceil\tau/2\rceil}$ entries, and each index into the table
requires $\lceil\tau/2\rceil$ bits. Note that larger values of $\tau$ result in
faster runtimes at the cost of increased space-usage.

In our application for sorting, $S$ is a sequence of $N$ bits containing
self-delimiting numbers. By this, each number in $S$ is of size $m = O(2^\bb)$.
Let $q = 2^{\bb / \tau}$ and call $x \in S$ {\em big} if $q < x$, otherwise call
$x$ {\em small}. We have to handle small and big numbers differently. To divide
the problem we scan through $S$ and write each small number of $S$ into a
sequence $S_{\le q}$ and each big number into a sequence $S_{> q}$. On the
word-RAM model, scanning through an $\bb$-bit sequence $S$ and reporting all $k$
numbers takes $O(k + \bb / \tau)$ time, which is the time bound of all our
sorting algorithms. After sorting both sequences we write the sorted numbers of
$S_{\le q}$, followed by the sorted numbers of $S_{> q}$, into a sequence $S'$ of $\bb$ bits.

{We first consider the sequence $S_{\le q}$.} Our idea is to run first an adaptation of
stable counting sort to {\em presort} the numbers in several {\em areas} such
that an area $A_i$ consists of all numbers that require exactly $i$ bits as
self-delimiting number. By doing so we roughly sort the sequence $S_{\le q}$ as 
all numbers of area $A_i$ are smaller than any number of area $A_j$ for all
$i<j$. We then sort each area independently by another stable counting-sort
algorithm.
\begin{lemma}\label{lem:ordering-self-delimiting} Given an $\bb$-bit sequence
	$S$ consisting of $k$ self-delimiting numbers, each in the range $[0, m)$ ($m
	\le 2^{\bb / \tau}$) and a parameter $\tau$ with $\log \bb \le \tau \le w$,
	{there is a
	stable-sorting algorithm computing a bit sequence of $\bb$ bits that stores the
	given self-delimiting numbers in sorted order.
	The algorithm runs in $O(k + \bb/\tau)+ o(2^{\tau})$ time and $O(\bb)+o(2^{\tau})$ bits.
	}
	\end{lemma}
	\begin{proof}
		We first describe how to presort the numbers into areas $A_i$.
		For an easier description, we define a function $f$ that maps each $x\in S$
			to its area (i.e., $x$ is in area $A_{f(x)}$ as follows:
		\begin{align*}
			f(x \in S) = 
			\begin{cases}
				 1 & x = 0\\
				 2\lfloor \log x \rfloor + 1 & x > 0
			 \end{cases}
		\end{align*}
		We start to reorder the numbers by copying them to a new, initially-empty
			 sequence $S'$ of $\bb$ bits to consecutively store the numbers of each
			 area $A_i$ ($i = [1, f(m)]$).
			 We need to assign to each non-empty area $A_i$ an offset $B(i) =
		\sum_{j \in [1, i - 1]} (j - 1)|A_j|$ in $S'$ such that every number of
		every area $j < i$ can be written in $S'$ by one scan over $S$.
		For the offset we require to know the cardinality $C(i) = |A_i|$ of every
		area $A_i$ ($i \in [1, f(m)]$). We realize $C$ as a table with~$f(m)$
		entries of $\lceil \log \bb \rceil$ bits each and use $C[i]$ to count how
		many numbers in $S$ require $i$ bits as a self-delimiting number. For that,
		iterate over the self-delimiting numbers in $S$ and for each number $x \in
		S$ increment $C[f(x)]$ by one. To be able to access the offset $B(i)$ of
		every area in constant time we also realize $B$ as an array of~$f(m)$ fields
		of $\lceil \log \bb \rceil$ bit each. Initialize $B$ by setting $B[1] = 1$
		and $B[i] = B[i - 1] + (i - 1)C[i - 1]$ with $i = 2, \ldots, f(x)$.
		After computing the offsets we use them to write the
		self-delimiting numbers into a new $\bb$-bit sequence $S'$ as follows: scan
		through $S$ and for each self-delimiting number $x$, store $x$ in $S'$
		starting from position $B[f(x)]$ and increment $B[f(x)]$ by~$f(x)$ to adjust
		the boundary for the next number. This completes the
		presorting of the numbers into a sequence $S'$. Note that
		the order $S'$ is stable.
	
		To sort the numbers inside an area we use a folklore technique. We consider
		each self-delimiting number of a number $x$ in an $2^{\tau/2}$-ary system so
		that it is represented by $\lceil f(x) / \lceil \tau /2 \rceil \rceil$
		digits from the alphabet $\mathcal{A}=[0,\lceil 2^{\tau/2}\rceil-1]$, and we
		use radix sort in combination with a stable counting sort
		algorithm. Consider the area $A_i$ as a sequence of $k_i$ subsequently
		stored self-delimiting numbers, each of $i$ bits and occupying $\bb_i$ bits
		in total. We sort the numbers digit-wise, starting with the least
		significant digit $d = 1, \ldots, \lceil 2i/\tau \rceil$. 
		We use a temporary sequence $S_i$ of $\bb_i$ bits to output the sorted
		numbers and subsequently replace $A_i$ in $S$ by $S_i$. We so sort every
		area $A_i$ in $S$.
		{ In detail, the folklore sorting algorithm consists of the
		following steps, repeating in rounds for each digit. Let $d$ be the
		current digit. Build a histogram of the occurrences of each element of
		the alphabet $\mathcal{A}$ at digit $d$. Following that, construct a
		prefix-sum array of the occurrences of each element of the alphabet at
		$d$, also known as the integral of the histogram. The prefix-sum array
		is then used to partially sort the numbers by the current digit.
		Concretely, if we have a number 
		with $d$th digit being 
		$\ell$, then a tuple $(d, \texttt{pre})$ that is stored at the
		$\ell$th position of the prefix-sum array 
		tells us to write  a number 
		at index $\texttt{pre}$ in the temporary sequence $S_i$.
		For each value that is placed in the temporary sequence $S_i$, the
		respective entry $(\ell, \texttt{pre})$ is replaced by $(\ell,
		\texttt{pre}-1)$. Once the numbers are partially sorted for the current
		digit $d$, continue with digit $d+1$ in the next round.
	
		We now describe how to obtain simultaneously the histogram and prefix-sum array for
		all rounds
		of the sorting algorithm for an area $A_i$. 
		We begin with a pre-processing step. Construct
		a table $T$ with $\lceil 2^{\tau/2} \rceil$ entries, one for each element
		of the alphabet $\mathcal{A}$. At each index $\ell \in \mathcal{A}$ of the table $T$ store
		an (initially empty) stack. Each stack at index $\ell$
		will store tuples $(d, \texttt{occ})$ where $\texttt{occ}$ is the numbers
		of occurrences
		of $\ell$ at digit $d$. Only tuples with $\texttt{occ}>0$
		are stored.
	
		For all digits $d$, iterate over all values stored at digit $d$ of
		the numbers of the current area $A_i$. For each element $\ell$ of the
		alphabet that occurs
		at digit $d$, increment the $\texttt{occ}$-value of the tuple on top of the
		stack that is stored at index $\ell$, or insert a $(d, 1)$ if it is the
		first occurrence of $\ell$ at digit $d$, indicated by a missing entry of
		$(d, \cdot)$ at the top of the stack. 
		After both iterations have finished, each stack is
		sorted such that, for two elements $e_1=(d_1, \texttt{occ}_1)$ and
		$e_2(d_2, \texttt{occ}_2)$ on the stack, it holds that $e_1$ comes before
		$e_2$ exactly if $d_1 < d_2$. Note that $d_1 \neq d_2$ for all distinct
		elements $e_1$ and $e_2$ on the stack. 
		
		Next, we show how to construct simultaneously the prefix-sum array for each digit. During
		the construction of the arrays, we initially maintain the prefix-sum arrays
		as lists. In detail, we maintain an array $P$ indexed by values $\ell \in \mathcal{A}$ that stores at
		$P[\ell]$ a list storing tuples $(d, \texttt{pre})$ representing the
		elements of the prefix-sum arrays. During the construction we store in a
		temporary array $P'$ for each digit $d$ the current prefix-sum while
		iterating over all table entries of $T$. Iterate over all entries of $T$,
		and when encountering a non-empty stack at $T[\ell]$, pop all elements of the stack. For
		each $(d, \texttt{occ})$ value removed, increment $P'[d]$ by
		$\texttt{occ}$, and append $(d, P'[d])$ to $P[\ell]$. 
	
		Once this procedure is finished, the prefix-sum array for each $d$ can be
		obtained via the lists stored in $P$. 
		Consider a round of the sorting algorithm where we sort by digit $d$.
		When encountering some $\ell \in \mathcal{A}$ during the iteration over the
		digits of all numbers of the current area, we look at $P[\ell]$
		to find the prefix-sum value $(d, \texttt{pre})$ as the first
		entry of the list stored at $P[\ell]$. As described, we replace
		$(d, \texttt{pre})$ with $(d, \texttt{pre}-1)$ and continue with the round.
		Once a round is finished, we remove the first element of each accessed
		list, i.e., all tuples $(d, \cdot)$. We so have the information
		for digit $d+1$ as the first entries of the lists in $P$ that
		are read during the sorting round for $d+1$.
	%
	
		We now analyze the space-bound and runtime of maintaining the tables and
		prefix-sum arrays over all areas. Note that there are $O(\log \bb)$ areas. For
		each area, we use $O(2^{\tau/2} \tau)$ bits to store the table $T$,
		excluding the space used for the stacks. The $k_i$ numbers of an area $A_i$
		contain $O(k_i+{\bb}_i/\tau)$ digits, and thus we can store the stacks of the
		table $T$ and the arrays $P$ and $P'$ with $O(k_1+N_1/\tau \cdot \tau+k_2 +
		N_2/\tau \cdot \tau+ \ldots)+o(2^{\tau})=O(N)+o(2^{\tau})$ bits. 
		We now analyze the runtime. For each area,
		we construct an empty table $T$ in time $O(2^{\tau/2})$ and scan over all
		digits of all numbers of the current area. The prefix-sum arrays of an area
		$A_i$ are constructed by another linear scan over the tables and digits,
		i.e., in time $O(k_i + {\bb}_i/\tau)$ time. Over all areas, the time spent on
		constructing the prefix-sum arrays is bound by $O(k+{\bb}/\tau)+o(2^{\tau})$.
	} 
	
		We now analyze the space-bound and runtime of the rest of the algorithm.
		The bit-vectors $B$ and $C$ use $O(f(m) \log \bb) =
		O(\bb/\tau \log \bb) = O(\bb/\log \bb \cdot \log \bb) = O(\bb)$
		bits and $S'$ uses $O(\bb)$ bits. 
		Scanning through~$S$, while computing and filling $C$ and $B$ can be done in
			$O(k + \bb/\tau) = O(k + \bb / \log \bb)$ time. (Roughly speaking, we
		can scan $S$ in $O(\bb/\tau)$ time and, to handle each number
		separately, costs us $O(k)$ extra time in total.)
		
		Recall that after presorting the numbers in areas, $A_i$ has $k_i$
		self-delimiting numbers stored subsequently in $\bb_i$ bits. Sorting $A_i$
		runs in $O(\bb_i/\tau + k_i + k_i\frac{i}{\tau})$ time since reading
		$A_i$ and reading every number runs in $O(\bb_i/\tau + k_i)$ time
		and the digit-wise sorting runs in $O(k_i\frac{i}{\tau})$ time. Over all
		areas the running time is {$O(\sum_{i \in [1,f(m)]} (\bb_i/\tau
		+ k_i + k_i\frac{i}{\tau})) = O(\sum_{i \in [1,f(m)]}(\bb_i/\tau + k_i)) +
		O(\sum_{i \in [1,f(m)]} k_i\frac{i}{\tau}) = O(\bb/\tau + k)$ time.}
	
		{In summary, the stable sorting algorithm runs in $O(\bb/\tau+k)+o(2^{\tau})$
		time and $O(\bb)+o(2^{\tau})$ bits.} \qed
	\end{proof}

{Now, let us look at the sequence $S_{> q}$}. The bound on the size of each
number implies that $S_{> q}$ cannot consist of more than $O(\bb/\log q) =
O(\tau)$ numbers and the biggest number may occupy $O(\bb)$ bits. The idea is to
interpret each number in $S_{> q}$ as a string using the
alphabet $\Sigma = \{0,\ldots, 2^{\tau/2}\}$.

Similarly, as in the previous lemma, we first sort the strings by their length
into areas such that each area consists of all self-delimiting numbers of one
length. Afterwards, we sort each area lexicographically using radix sort. Note
that directly lexicographically sorting binary sequences when interpreted as
strings does not result in a sorted sequence of self-delimiting numbers. For
example $0, 11, 1010, 1011$ are sorted self-delimiting numbers, but directly
interpreting them as strings would sort them lexicographically as $0, 1010,
1011, 11$.

\begin{theorem}\label{thm:sdn-sorting} Given an $\bb$-bit sequence $S$ of $k$
	self-delimiting numbers and a parameter $\tau$ with $\log \bb \le \tau \le
    w$ {
	there is a stable-sorting algorithm computing a bit sequence of $\bb$ bits that stores
	the given self-delimiting numbers in sorted order.
    The algorithm runs in time $O(k+\bb/\tau)+o(2^{\tau})$ and $O(\bb)+o(2^{\tau})$ bits.}
\end{theorem}
\begin{proof}
The idea is to put the numbers of $S$ that are at most $q = 2^{\bb / \tau}$ into
a sequence $S_{\le q}$ and sort it by using
Lemma~\ref{lem:ordering-self-delimiting}. We write the remaining
$O(\frac{\bb}{\log q}) = O(\tau)$ numbers into a sequence $S_{> q}$ where we
interpret each number as a string of digits out of the alphabet $\Sigma = \{0,
\ldots, 2^{\tau/2}\}$ stored as self-delimiting numbers. To sort
$S_{> q}$ we first sort the numbers by their length in bits into areas and then
sort the areas independently. The difference to the previous lemma is that the
numbers are big, and thus we cannot move a number in constant time. The new
approach comes from string sorting. We create a sequence of pairs, where the
$i$th pair consists of the length in bits and the position of the $i$th number
of $S_{> q}$, respectively. Instead of sorting the numbers, we sort the pairs by
their first component (i.e., by the length in bits) utilizing counting sort. By
sequentially reading the pairs we can access each number in $S_{> q}$ using the
second component of the pair and write all numbers into a sequence $S'_{> q}$ in
order of their length. Now we can sort the numbers area-wise employing
radix-sort. Assume that we sort area $A_i$ that consists of numbers that all
consist of $i$ bits. Again, instead of reading the whole self-delimiting number,
we interpret them as strings of $\lceil i/\lceil \tau /2 \rceil \rceil$ digits
and sort the tuples by the $d = 1, \ldots, \lceil i/\lceil \tau /2 \rceil
\rceil$ digit of the numbers, starting with the least significant digit.
After sorting $S_{\le q}$ and $S_{> q}$ we can write them into a single
$\bb$-bit sequence and return it as the result.

We now consider the space consumption of the algorithm above. The sequence of
tuples uses $O(\tau \log \bb)$ bits. During the algorithm we write the
self-delimiting numbers in $\bb$-bit sequences. 
The application of counting sort on digits of $\lceil \tau/2 \rceil$ bits
requires tables with $2^{\lceil \tau/2 \rceil}$ entries, each of
$O(\tau)$ bits. The tables use $o(2^{\tau})$ bits.
Together with the application of Lemma~\ref{lem:ordering-self-delimiting} our
algorithm uses {$O(\bb) +o(2^{\tau})$ bits in total.}
 
 Sorting the numbers by their length
runs in $O(\tau)$ time using counting sort. Sorting an $\bb_i$ bit area storing
$k_i$ self-delimiting numbers, each of $i$ bits, runs in $O(k_i i/\tau)$ time.
Over all existing areas the algorithm runs in $O(\sum_i k_i i/\tau) = O(k +
\bb/\tau)$ time. {Together with Lemma~\ref{lem:ordering-self-delimiting}
this results in a runtime of $O(k+N/\tau)  + o(2^{\tau})$.} \qed
\end{proof}

We now consider dense/competitive ranking of $k$ self-delimiting numbers
of the universe $\{0, \ldots, m\}$
stored as an $\bb$-bit sequence $S$. 
A standard approach to compute the dense/competitive rank is to first sort~$S$
and then to use an array~$P$ of~$m$ entries, each of $\lceil \log k \rceil$
bits, to store a prefix sum over the occurrences of (different) numbers~$x \in
S$, i.e., in a first step for competitive rank set $P[x] = P[x] + 1$ (for dense
rank, set $P[x] = 1$) for each $x \in S$. In a second step compute the prefix
sums on $P$, i.e., for each $i = 1, \ldots, m - 2$, set $P[i] = P[i - 1] +
P[i]$. The dense/competitive rank of a number $x$ is then $P[x]$. However, the
array $P$ requires $\Theta(m \log k)$ bits. To compute the dense rank with
less space, we can use a bit-vector $B$ of $m$ bits and set $B[x] = 1$ for each
$x \in S$. Then, using a rank-select data structure on $B$, the dense rank of
$x$ is $\op{rank}_B(x)$. This approach uses $O(m)$ bits and takes
{$\Theta(m/w)$} time due to the  
initialization of a rank-select data structure on~$B$~\cite{BauH17}. Note that
this solution does not handle duplicates, due to the use of a bit-vector as the
underlying data structure.

Our novel approach provides a dense/competitive rank solution that does not rely
on the universe size $m$ in both the runtime (for initialization) and bits
required, but only on $N$. Moreover, for our use-case in
Section~\ref{sec:treeiso} we use dense rank with universe size $m=O(2^N)$ for
which the approaches outlined previously require $\Omega(N \log N)$ bits, while
we aim for $O(N)$ bits. Due to this specific requirement we require a novel
approach, which works similar to the realization of dense rank
in~\cite{BauH17,Elias74,Jac89}.

We first discuss dense rank. Similar to our solution for sorting, we handle
small and large numbers differently. Let $S$ be an $\bb$-bit sequence of $k$
self-delimiting numbers. Denote with $S_{\leq \bb}$ the subsequence of $S$ that
contains numbers that are at most size $N$, and with $S_{> \bb}$ all numbers of
$S$ that are larger than $N$ (and at most size $2^N$). We first discuss the
techniques used for enabling dense rank for all small numbers $S_{\leq \bb}$ of
$S$ for which we build our own data structure. Denote with $m'$ the size of the
universe of $S_{\leq \bb}$ and $k'$ the number of self-delimiting numbers
contained in $S_{\leq \bb}$. We construct the dense rank structure not on a
given bit-vector, but on a given sequence consisting of $k'$ self-delimiting
numbers that correspond to the ones in the bit-vector. For this, we construct a
bit vector $B$, which we partition into $O(m'/\tau)$ frames of $\lceil \tau/2
\rceil$ bits and create an array $P$ that contains the prefix sum of the frames
up to the $i$th frame ($i = 0, \ldots, \lceil m' / \lceil \tau /2 \rceil
\rceil$). For an example, see Fig.~\ref{fig:dense-rank-schema}. 
Subsequently, we
use a lookup table that allows to determine the number of ones in the binary
representation of each frame. The table $\op{POPCNT}$ can be easily
constructed in time $O(2^{\tau/2}\tau) = o(2^\tau)$ as follows: for each $z = 0,
\ldots, 2^{\tau/2} - 1$ set $\op{POPCNT}[z] = y$, where $y$ is the number of
bits for $z$ computed in $O(\tau)$ time by iterating over the bits of $z$'s
binary representation.

\begin{figure}[h]
	\centering
	\begin{tikzpicture}[
every node/.style={minimum width=1.32em, minimum height=1.2em, inner sep=0pt,
align=center, fill=white, font=\small},
node distance = 0pt,
unused/.style={pattern=north east lines, pattern color=gray}]

\begin{scope}[name=s1, start chain = going right]%
\node [box, minimum width=0.6cm] (1) {1};
\node [box, minimum width=0.6cm] (2) {1};
\node [minimum width=1.32em, on chain] (dd) {$\ldots$};
\node [box, minimum width=1.1cm] (2b) {64};
\node [box, minimum width=1.1cm] (3) {64};
\node [box, minimum width=1.1cm] (4) {78};
\node [box, minimum width=1.1cm] (5) {99};
\node [box, minimum width=1.1cm] (6) {260};
\node [box, minimum width=1.4cm] (7) {467};
\node [box, minimum width=1.45cm] (8) {3004};
\node [box, minimum width=1.45cm] (9) {3044};
\node [on chain, left=0 of 1] (S) {$S_{\leq \bb}$};
\end{scope}

\def\x{3.63em}

\begin{scope}[start chain = going right, shift={($(1.west)+(0,-0.9cm)$)},anchor=west]%
\node[box, minimum width=\x, fill=white] (q) {1000000};
\node [box, minimum width=\x, unused] (b1) {};
\node[box, minimum width=\x, fill=white] (q2) {1001001};
\node [box, minimum width=\x, unused] (b2) {};
\node [box, minimum width=\x, unused] (b3) {};
\node[box, minimum width=\x, fill=white] (q3) {0100010};
\node [box, minimum width=\x, unused] (b4) {};
\node [box, minimum width=\x, unused] (b5) {};
\node[box, minimum width=\x, fill=white] (q4) {0100100};
\node [on chain, left=0 of q] (B) {$B$};
\end{scope}

\begin{scope}[start chain = going right, shift={($(1.west)+(0,-1.4cm)$)}, anchor=west]%
\node [box, minimum width=\x] (p1) {0};
\node [box, minimum width=\x, unused] (p2) {};
\node [box, minimum width=\x] (p3) {1};
\node [box, minimum width=\x, unused] (p4) {};
\node [box, minimum width=\x, unused] (p5) {};
\node [box, minimum width=\x] (p6) {4};
\node [box, minimum width=\x, unused] (p7) {};
\node [box, minimum width=\x, unused] (p8) {};
\node [box, minimum width=\x] (p9) {6};
\node [on chain, left=0 of p1] (P) {$P$};
\end{scope}

\draw[->] (1) to[out=-90, in=90] ($(p1.west) +  (0.45em, 2.1em)$);
\draw[->] (2) to[out=-90, in=90] ($(p1.west) + (0.45em, 2.1em)$);

\draw[->] (2b) to[out=-90, in=90] ($(q2.west) + (0.45em,0.6em)$);
\draw[->] (3) to[out=-90, in=90] ($(q2.west) +  (0.45em,0.6em)$);
\draw[->] (4) to[out=-130, in=90] ($(q2.west) +  (1.8em,0.6em)$);
\draw[->] (5) to[out=-90, in=90] ($(q2.west) +  (3.2em,0.6em)$);

\draw[->] (6) to[out=-90, in=90] ($(q3.west) +  (0.9em,0.6em)$);

\draw[->] (7) to[out=-90, in=90] ($(q3.west) +  (2.8em,0.6em)$);

\draw[->] (8) to[out=-90, in=90] ($(q4.west) +  (0.9em,0.6em)$);
\draw[->] (9) to[out=-90, in=90] ($(q4.west) + (2.3em,0.6em)$);

\end{tikzpicture}		
	\caption{A sketch of our storage scheme to realize dense rank. {$S_{\leq
	\bb}$} is a sorted sequence of self-delimiting numbers, $B$ and $P$ are
	vectors partitioned into frames. For each $x \in S_{\leq \bb}$ we flip the
	$x$th bit in $B$. Afterwards we count the number of bits for each frame and
	store the prefix sum over the numbers in~$P$. The gray stripped frames are
	uninitialized because no number of $S_{\leq \bb}$ belongs to that
	frame.}\label{fig:dense-rank-schema}
\end{figure}

It remains to show a solution for big numbers $S_{> \bb}$. Note that the dense
rank of any number cannot be bigger than $\bb$ and thus use more than $\log \bb$
bits. On the other hand, $S_{> \bb}$ contains $\leq N / \log N$ numbers. Thus,
we can use an $\bb$-bit vector $Q$ consisting at most $\leq N/\log N$ entries,
each of $\log N$ bits, and if (intuitively speaking) drawing the vector below
$S_{> \bb}$, we can write the dense rank of every number $x \in S_{> \bb}$ with
$\bb < x \le 2^\bb$ into $Q$ below $x$. By requiring that the access to the
dense rank of $x \in S_{> \bb}$ has to be done by providing the position $p_x$
of the first bit of $x$ in $S$ as the input (instead of the number $x$ itself),
we can easily return the dense rank of $x$ in constant time. Note that we need
the position $p_x$ since the binary representation of $p_x$ can be always be
written with $\log N \leq w$ bits, but the binary representation
of $x$ can not. This allows constant time operations.

\begin{theorem}\label{thm:dense-rank} {Given an $\bb$-bit sequence
	$S$ of $k$ self-delimiting numbers, a parameter $\tau$ with $\log \bb \le r
	\le w$ we can compute a data structure realizing dense rank. The data
	structure can be constructed in $O(k + \bb / \tau)+o(2^\tau) = O(k + \bb /
	\log \bb)+o(2^\tau)$ time and uses $O(\bb) + o(2^\tau)$ bits. For numbers
	$x$ of size $>\bb$ the position $p_x$ of $x$ in $S$ is required as the
	input.}
\end{theorem}
\begin{proof}
{In the following we use a lookup table $\op{POPCNT}$ that allows
to query the number of ones in the binary representation of the elements for
which the table was constructed. We construct the table for binary
representations of at most ${\tau/2}$ bits. The table can be easily constructed
in time $O(2^{\tau/2}\tau) = o(2^\tau)$ as follows: for each $z = 0, \ldots,
2^{\tau/2} - 1$ set $\op{POPCNT}[z] = y$, where $y$ is the number of bits for
$z$ computed in $O(\tau)$ time by iterating over the bits of $z$'s binary
representation. The table requires $o(2^\tau)$ time to construct.}

We split the sequence~$S$ of self-delimiting numbers by putting each number~$x
\in S$ with~$x \le \bb$ into a sequence~$S_{\le \bb}$ and each number~$x \in S$
with~$x > \bb$ into a sequence~$S_{> \bb}$. Determine the dense rank for
elements in the sequence~$S_{\le \bb}$ as described in the paragraphs below. For
$S_{> \bb}$ proceed as follows. Create a bit vector $Q$ of length $O(\bb)$
initialized with zeros in constant time~\cite{DBLP:journals/tcs/KatohG22}.
Intuitively speaking, if $Q$ is drawn below $S$ we write the dense rank of a
number $x \in S_{> \bb}$ below~$x \in S$ in $Q$. In detail, we first write $p_x$
of every number $x \in S$ into a bit sequence $I$ and our sorting algorithm
sorts $S_{> \bb}$ such that it also moves the positions $p_x \in I$ with the
numbers in $S_{> \bb}$. Now we determine the dense rank for each number $x \in
S_{> \bb}$ and use the position $p_x$ to write the dense rank in an $\bb$ bit
vector $Q$ at position $p_x$.
Note that the dense rank of the numbers and their positions are at most $\bb$ so
 that storing them uses less space than the numbers themselves. In other words,
 the space for storing several dense ranks and positions in $Q$ and $I$,
 respectively, does not overlap. If afterwards a dense rank is queried, check
 first if it is a number greater than $\bb$. If so, use $Q$ to answer the query
 and, else, use the description below.

Our solution for numbers of {$S_{\leq \bb}$} follows. We start to
describe some auxiliary data structures. Let $B$ be a bit-vector of $\bb$ bits
containing the distinct values of~$S_{\leq \bb}$, i.e., $B[x] = 1$ for
every $x \in S_{\leq \bb}$. Partition $B$ into frames of $\lceil \tau/2 \rceil$
bits each. Let $P$ be an array of $\lceil \bb/(\tau/2)\rceil +1$ entries, each
of $\lceil \tau/2 \rceil$ bits, that allows us to store a prefix sum over the
bits inside the frames of $B$, i.e., $P[0] = 0$ and $P[i + 1]$ ($i = 0, \ldots,
\lceil \bb/\lceil \tau /2 \rceil \rceil$) is the number of ones within $B[0,
\ldots, i\tau/2]$.

We now show how to construct $B$ and $P$. {Apply
Theorem~\ref{thm:sdn-sorting} on $S_{\leq \bb}$ to get a sorted sequence
$S'_{\leq \bb} = \{x_0, \ldots, x_{k - 1}\}$.} Let $j$ (initially $j = 0$) be a
counter that we use to count different numbers in $S'_{\leq \bb}$. Initially set
$P[0] = 0$. Now for each $x_i$ with $i = 0, \ldots, k - 1$ do the following: if
and only if $B[x_i] = 0$, increment $j$ by one. In any case, set $B[x_i] = 1$
and $P[\lfloor x_i / \lceil \tau /2 \rceil \rfloor + 1] = j$.

We can answer the dense rank of a number $x \in S_{\leq \bb}$ by
 returning $P[q] + \op{POPCNT}[B[q\tau/2, \ldots, (q + 1)\lceil \tau /2 \rceil -
 1] \& (1 \ll (\lceil \tau /2 \rceil - p)) - 1]$ with $q = \lfloor x / \lceil
 \tau /2 \rceil \rfloor - 1$ and $p = x \mod \lceil \tau/2 \rceil$.
Initializing the memory of $B$ and $P$ can be done in constant time. Sorting $S$
(Theorem~\ref{thm:sdn-sorting}) and computing $B$ and $P$ can be done in $O(k +
\bb / \tau)$ time. Both, $B$ and $P$ use $O(\bb)$ bits. The lookup table
requires $O(2^{\tau})$ bits. In total, we use $O(\bb)$ bits.  
\end{proof}

To compute the competitive rank we require the information of how many times an
element appears in the given sequence. We change our approach of the previous
lemma as follows: recall that $\tau$ is a {parameter} with $\log
\bb \le \tau \le w$. We sort~$S$ to get a sorted sequence~$S'$. Next, we
partition $S'$ into {\em regions} such that the $i$th region $\mathcal R_i = S'
\cap [i\lceil\tau/2\rceil, \ldots, (i + 1)\lceil \tau / 2 \rceil]$ for all $i =
0, \ldots, 2 \lceil m / \lceil \tau /2 \rceil \rceil - 1$. In detail, we go
through $S'$ and store for each non-empty region $\mathcal R_i$ a pointer $F[i]$
to a sequence $A_i$ of occurrences of each number $x \in S'$ written as
self-delimiting numbers. Similar to the usage of $B$ for dense rank  (Fig.~\ref{fig:dense-rank-schema}), 
we solve competitive rank by partitioning
$A_i$ into frames of~$\lceil \tau/2 \rceil$ bits and computing an array $P_i$
storing the prefix-sums. Inside a single frame we use a lookup table to compute
the prefix sum. More exactly, $P_i[j]$ stores the prefix-sum over all
self-delimiting numbers in $S'$ up to the $j$th frame in $A_i$.
Fig.~\ref{fig:rank-schema} sketches an example. We so obtain the next theorem
where we again require the position $p_x$ in $S$ of all $x>\bb$ as the
input to our competitive rank query.
 
\begin{figure}[h]
	\centering
	\begin{tikzpicture}[
every node/.style={minimum width=1.32em, minimum height=1.2em, inner sep=0pt,
align=center, fill=white, font=\small},
node distance = 0pt,
unused/.style={pattern=north east lines, pattern color=gray}]
\begin{scope}[name=scope1, start chain = going right]%
\node [box, minimum width=11cm] (1) {
$\langle 2 \rangle$$\langle 2 \rangle$$\langle 2 \rangle$$\langle 2 \rangle$$\langle 2 \rangle$$\langle 2 \rangle$
$\langle 5 \rangle$$\langle 5 \rangle$$\langle 5 \rangle$
$\langle 6 \rangle$
$\langle 8 \rangle$$\langle 8 \rangle$$\langle 8 \rangle$
$\langle 256 \rangle$$\langle 300 \rangle$$\langle 300 \rangle$$\langle 300 \rangle$
$\langle 1012 \rangle$
$\langle 1012 \rangle$
};

\node[shift={($(1.west) + (-0.6em,1.3em)$)}, fill=none] (r1) {};
\node[shift={($(1.west) + (17.9em,1.3em)$)}, fill=none] (r2) {};
\node[shift={($(1.west) + (30.2em,1.3em)$)}, fill=none] (r3) {};
\node[shift={($(1.east) + (0.7em,1.3em)$)}, fill=none] (r4) {};
\path[draw] (r1) -- node[midway] {$\mathcal R_1$} ($(r2) - (2pt, 0)$);
\path[draw] ($(r2) + (2pt, 0)$) -- node[midway] {$\mathcal R_8$} ($(r3) - (27pt, 0)$);
\path[draw] ($(r3) - (24pt, 0)$) -- node[midway] {$\mathcal R_{10}$} (r4);

\node [on chain, left=0 of 1] (S) {$S'$};
\end{scope}

\begin{scope}[name=s2, start chain = going right, shift={($(1.west)+(0,-1.1cm)$)},anchor=west]%
\node [draw, minimum width=11.6cm, unused] (b) {};
    
\node[draw, minimum width=3.96em, fill=white] (q) {};
\node[draw, minimum width=3.96em, fill=white, right=15em of q] (q2) {};
\node[draw, minimum width=3.96em, fill=white, right=4.8em of q2] (q3) {};
    
\node [above=0 of q, fill = none] () {1};
\node [above=0 of q2, fill = none] () {8};  
\node [above=0 of q3, fill = none] () {10};  
    
\node [on chain, left=0 of b] (B) {$F$};
\end{scope}

\begin{scope}[start chain = going right, shift={($(1.west)+(0,-1.9cm)$)}, anchor=west]%
\node [box, minimum width=3.96em] (a1) {
$\langle 6 \rangle$$\langle 3\rangle$
};
\node [box, minimum width=3.96em] (a2) {
$\langle 1 \rangle$$\langle 3 \rangle$
};
\node [on chain, left=0 of a1] (P) {$A_1$};
\end{scope}

\begin{scope}[start chain = going right, shift={($(1.west)+(0,-2.4cm)$)}, anchor=west]%
\node [box, minimum width=3.96em] (p1) {0};
\node [box, minimum width=3.96em] (p2) {9};
\node [on chain, left=0 of p1] (P) {$P_1$};
\end{scope}

\begin{scope}[start chain = going right, shift={($(1.west)+(19em,-1.9cm)$)}, anchor=west]%
\node [box, minimum width=3.96em] (a2) {
$\langle 1 \rangle$$\langle 3 \rangle$
};
\node [on chain, left=0 of a2] (P) {$A_8$};
\end{scope}

\begin{scope}[start chain = going right, shift={($(1.west)+(19em,-2.4cm)$)}, anchor=west]%
\node [box, minimum width=3.96em] (p1) {14};
\node [on chain, left=0 of p1] (P) {$P_8$};
\end{scope}

\begin{scope}[start chain = going right, shift={($(1.west)+(27.75em,-1.9cm)$)}, anchor=west]%
\node [box, minimum width=3.96em] (a3) {
$\langle 2 \rangle$
};
\node [on chain, left=0 of a3] (P) {$A_{10}$};
\end{scope}

\begin{scope}[start chain = going right, shift={($(1.west)+(27.75em,-2.4cm)$)}, anchor=west]%
\node [box, minimum width=3.96em] (p1) {18};
\node [on chain, left=0 of p1] (P) {$P_{10}$};
\end{scope}

\draw[->] (q.center) to[out=-90, in=90] ($(a1.north) +  (0em,0)$);
\draw[->] (q2.center) to[out=-90, in=90] ($(a2.north) +  (0em,0)$);
\draw[->] (q3.center) to[out=-90, in=90] ($(a3.north) +  (0em,0)$);

\end{tikzpicture}		
	\caption{A sketch of our storage schema to realize competitive rank. For
	each region $\mathcal R_i$, that contains numbers out of $S'$, a pointer
	$F[i]$ points to a data structure storing the amount of occurrences for each
	of the numbers. In addition, $P_i$ stores the prefix-sum over the frames up
	to $A_i$. For numbers $x$ of size $>\bb$ the position $p_x$ of $x$ in $S$ is
	required as the input.}\label{fig:rank-schema}
\end{figure}
	
\begin{theorem}\label{th:rank} {Given an $\bb$-bit sequence $S$ of $k$
	self-delimiting numbers, a parameter $\tau$ with $\log \bb \le \tau \le w$
	we can compute a data structure realizing competitive rank. The data
	structure can be constructed in $O(k + \bb / \tau) + o(2^{\tau}) = O(k + \bb
	/\log \bb) + o(2^{\tau})$ time and uses $O(\bb) + o(2^{\tau})$ bits. }
\end{theorem}
\begin{proof}
We make use of a table $\op{PREFIXSUM}$ that stores for all
binary sequences of length $\lceil \tau/2 \rceil$ the answer to all prefix-sum queries.
The table can easily be computed in time $o(2^{\tau})$ and stored with
$o(2^{\tau})$ bits.

We can handle numbers in the range of $[\bb, 2^{\bb})$ the same way we did in
Theorem~\ref{thm:dense-rank}, that is, storing the competitive rank of a number
$x$ in an extra  $O(\bb)$-bit sequence $I$ starting at bit $p_x$, where $p_x$ is the
position of $x$ in $S$. Thus, it suffices to consider the case that the given
self-delimiting numbers are all in the range $[0, \bb)$.

	We start to describe our storage schema. Partition $[0, \bb)$ into $\lceil
	\bb/\tau \rceil$ regions where each region $\mathcal R_i = [i \lceil \tau /
	2 \rceil, (i + 1)\lceil \tau/2 \rceil)$ (with $i = 0, \ldots, 2 \lceil \bb /
	\lceil \tau/2 \rceil \rceil - 1$) consists of $\lceil \tau/2 \rceil$
	different numbers.
	Create an array $B$ of $\lceil \bb / \lceil \tau / 2\rceil \rceil$ fields,
	each of $\lceil \tau/2 \rceil$ bits. For each region $\mathcal R_i$ with
	$\mathcal R_i \cap S = \emptyset$, $B[i]$ contains a \op{null} pointer.
	Otherwise, $B[i]$ points to a data structure $Q_i$ that contains all
	information to determine the competitive rank of each number of $\mathcal
	R_i \cap S$.
	
	The data structure $Q_i$ consists of two auxiliary structures $A_i$ and
	$P_i$ defined in the next paragraphs. Let $A_i$ be a sequence of
	self-delimiting numbers consisting of the number of occurrences of each
	number $x \in \mathcal R_i$ in $S$, stored via
	static space allocation (Section~\ref{sec:preone}). Note that, if a number $x \in \mathcal
	R_i$ does not appear in $S$, we store $0$ as the number of its occurrences
	in $S$. Partition $A_i$ into frames of $\lceil \tau/2 \rceil$ bits each.
	
	A frame can contain multiple self-delimiting numbers. If this is the case,
	we want to access the prefix sum over the self-delimiting numbers inside the
	frame, up to a desired number, quickly. To do so we use the precomputed
	lookup table $\op{PREFIXSUM}$. However, to simplify the
	process we want no self-delimiting number to be divided between two frames.
	Therefore, we store each self-delimiting number in $A_i$ so that it either
	fits into the current not full frame, or it uses the next one. Of course it
	is not possible to store a large number (of more than $\lceil \tau/2 \rceil$
	bits) into a single frame. In this case a frame contains only one number and
	we can read it directly without using table lookup. If the number is too
	large to fit into a single frame, we use the next three frames to store the
	number in a standard binary representation. In detail, use the next free
	frame to store only ones and so mark the subsequent two frames to store the
	number in a binary representation instead of a self-delimiting number.

	To find the index of the first bit of the number of
    $\mathcal R_i$ inside $A_i$, we use static space allocation, as described in
    Section~\ref{sec:preone}. Moreover, we use an array $P_i$ of the same size
    as $A_i$ partitioned into frames of size $\lceil \tau/2 \rceil$, where
    $P_i[\ell]$ is the prefix-sum up to the $\ell$th frame of $A_i$, including
    the prefix sum of the numbers in the regions $0 \ldots i - 1$. For an
    example, see again Fig.~\ref{fig:rank-schema}.

	With the data structures defined above we can answer the competitive rank of
	a number $x \in S$ as follows: compute the region $i = \lfloor x / \lceil
	\tau /2 \rceil \rfloor$ of $x$. Use the pointer $F[i]$ to jump to a data
	structure $Q_i$. The number $x$ is the $q$th number inside the region
	$\mathcal R_i$ with $q = x \mod (\tau / 2)$. Select the
	starting position $p$ of the $q$th number in $A_i$, accessible via static
	space allocation. The position $p$ lies
	inside the $f$th frame with $f = \lfloor p / \lceil \tau /2 \rceil \rfloor$.
	Read the $f$th frame as a variable $z$ and check if the self-delimiting
	number $A_i[q]$ occupies the whole frame, i.e., check if $A_i[q]$ encodes
	the start of a large number. If not, remove all numbers inside $z$ after the
	number $A_i[q]$ and return the competitive rank of $x$ as $P[f] +
	\op{PREFIXSUM}[z]$. Otherwise, read the large number as $z$ and return the
	competitive rank of $x$ as $P[f] + z$

	We now describe the construction of our data structure. Apply
	Theorem~\ref{thm:sdn-sorting} on $S$ to get a sorted sequence $S'$.
	Initialize $F$ with $\lceil \bb / \lceil \tau /2 \rceil \rceil$
	zero-entries, each of $\tau/2$ bits, and let $j = 0$ be the largest
	competitive rank so far. Iterate over $S'$ and for each number $x$ check if
	$F[i]$ with $i = \lfloor x / \lceil \tau /2 \rceil \rfloor$contains a
	pointer to a data structure $Q_i$. If not, initialize $Q_i$ as follows: as
	long as the numbers are part of the same region $\mathcal R_i$, run through
	$S'$, and count in $y$ the different numbers in that region. Since $\mathcal
	R_i$ contains $y \le \lceil \bb/r \rceil$ different numbers out of $S$, we
	can temporarily afford to use a key value map $D$ of $\lceil y/\lceil \tau
	/2 \rceil \rceil$ entries, each of $\tau$ bits, to count the number of
	occurrences of each number $x \in \mathcal R_i$ in $S$.
	Let $\#(x, S)$ be the number of occurrences of a number $x$ in $S$.
	
	In detail as long as the next number in $S'$ is part of $\mathcal R_i$, run
        through $S'$ and, for each number $x$ found in $S'$, store the number of
        its occurrences $\#(x, S)$ as $D[x] = \#(x, S)$. Afterwards, scan
        through $D$ and determine the number $\bb'$ of bits required to
        represent all  numbers in $D$ as self-delimiting numbers. Create $Q_i$,
        i.e., allocate $2\bb' + \tau/2$ bits for the static space allocation of
        $A_i$ as well as create an array with $\lceil (3\bb' + \tau/2)/\lceil
        \tau /2 \rceil \rceil$ entries of $r$ bits for $P_i$. Fill $A_i$ with
        the numbers of $D$ as follows: take the difference $d$ between the
        number $z$ (initially 0) and the next key number $z'$ of $D$, jump over
        $d$ bits of $A_i$ and store the number $D[z']$ in $A_i$ such that a
        number either fits into the current frame of $A_i$, the next, or if the
        number is too large to fit into a frame, store it as described above.
        Moreover, sum up all written numbers in a global variable $j$ (initially
        0), i.e., set $j = j + D[z']$, which is the largest prefix sum so far.
        Use it to fill $P_i$ every time a frame $\ell$ is about to be written,
        i.e., set $P_i[\ell] = j$. If the iteration over the numbers of the
        region $\mathcal R_i$ are finished in $S'$, delete $D$ and continue the
        iteration over $S'$.
	
	We now analyze the space complexity of the data structures. The array $F$
	and the sequence $S'$ use $O(\bb)$ bits. Recall that $F[i]$ contains a
	\op{null} pointer if $S$ does not contain any number of region $\mathcal
	R_i$. Otherwise, a structure $Q_i$ is created. We now focus on the space
	complexity of $Q_i$. We store in $A_i$ only the amount of occurrences of
	each different number of $\mathcal R_i$ in $S$. Observe that, if a number
	occurs $z$ times, then $O(\log z)$ bits suffice to store the amount of
	occurrences as a self-delimiting number (each occurrence contributes $O(1)$
	bits). However, some numbers of the region $\mathcal R_i$ do not occur, but
	we represent them anyway as a $0$ in $A_i$, which takes us $O(1)$ bits for
	every such number. Based on this observation $A_i$ uses $(\sum_{x\in
	\mathcal R_i \cap S} O(\log \#(x, S))) + O(\tau)$ bits with $\#(x, S)$ being
	the number of occurrences of $x$ in $S$. In total, $Q_i$ uses $(\sum_{x\in
	\mathcal R_i \cap S} O(\log \#(x, S))) + O(r) = O(d_i) + O(\tau)$ bits with
	$d_i$ being the amount of all numbers in $S$ that are part of $\mathcal
	R_i$.

	We create $Q_i$ only if $S$ contains numbers of a region $\mathcal R_i$.
	Summed over the $k \le N/\tau$ non-empty regions, our structure uses
	$(\sum_{x = 0}^{k} O(\log \#(x, S))) + O(\tau k) = \sum_{x = 0}^{k} O(\log
	\#(x, S)) + O(\bb) = O(\bb)$ bits.
		
	For a structure $Q_i$, the array $A_i$ can be constructed in
	$O(k')$ time, where $k'$ are the number of self-delimiting numbers in
	$\mathcal R_i \cap S$. Let the numbers inside the $i$th region occupy $\bb'$
	bits. The rank-select structure can be initialized in $O(\bb'/\tau)$ time.
	Since our smallest possible self-delimiting number uses $\Omega(1)$ bits and
	the largest possible uses $O(\tau)$ bits, the sum of the numbers of region
	$\mathcal R_i$ is within $\bb'$ and $2^{\bb'}$ and our amount of numbers
	within $\bb'$ bits is within the bounds $\Omega(\bb' / \tau)$ (if most
	self-delimiting numbers are large) and $O(\bb')$ (if most self-delimiting
	numbers are small). In the worst-case, we store $\Theta(\bb'/\tau)$ numbers
	within the $\bb'$ bits, and we end up with a
        construction time that is linear to the amount of numbers in $Q_i$.
	Therefore, in total our structure can be constructed in $O(k + \bb / \tau)$
	time. 
\end{proof}

\section{Preliminaries for Tree Isomorphism}\label{sec:prelim} In this paper we
use basic graph and tree terminology as given in~\cite{CorLRS09}. By
designating a {node} of a tree as the root, the tree becomes a {\em rooted}
tree. If the nodes of a tree have labels, the tree is {\em labeled}, otherwise,
the tree is called {\em unlabeled}. If there is a color
associated with each node of the tree, it is called \textit{node-colored}. {The
\textit{parent} of a node in a rooted tree is the neighbor of~$u$ on the
shortest path from~$u$ to the root. The root has no parent. The
\textit{children} of a node~$u$ are the neighbors of~$u$ that are not its
parent. A \textit{leaf} is a node with no children. Two nodes that share
the same parent are \textit{siblings}. A \textit{descendant} of a node is
either a child of the node or a child of some descendant of the node. By
fixing the order of the children of each node a rooted tree becomes an {\em
ordinal} tree. The \textit{right sibling} (\textit{left sibling}) of a node
$u$ is the sibling of~$u$ that comes after (before)~$u$ in the aforementioned
order, if it exists. } We denote by $\op{deg}(v)$ the degree of a node $v$,
i.e., the number of neighbors of~$v$, and by $\op{desc}(u)$ the number of
descendants of a node $u$ in a tree. The {\em height} of~$u$ is defined as the
number of edges between~$u$ and the longest path to a descendant leaf. The {\em
depth} of~$u$ is defined as the number of edges on the path between~$u$ and the
root.

\begin{definition}{(balanced parenthesis representation of an ordinal tree)}
A balanced parenthesis representation $\op{bpr}$ of an ordinal $n$-node tree~$T$
is a sequence of $2n$ parenthesis and is defined recursively on the root $u$ of
$T$ as 
\begin{align*}
\op{bpr}_u = {(}\ \op{bpr}_{v_1} \ldots \op{bpr}_{v_{\op{deg}(u)}} \ {)}	
\end{align*}
where $v_1$ is the leftmost child of $u$ and the $v_{\op{deg}(u)}$ the rightmost
child of $u$.
\end{definition}
For this representation, Munro and Raman~\cite{MunR97} showed a succinct data
structure {(using $2n + o(n)$ bits)} that allows constant-time {\em tree
navigation}, which we use. A similar data structure was previously shown by
Jacobson~\cite{Jac89} who showed a succinct data structure called \textit{level
order unary degree sequence} (LOUDS) that also provides tree navigation.

\begin{definition}{(tree navigation)}
Given a node $u$ of an ordinal tree $T$, tree navigation allows access to
\begin{itemize}
	\item $\op{parent}(u)$ that returns the parent of $u$, or $\op{null}$ if $u$
	is the root,
	\item $\op{firstChild}(u)$ returns the first child of $u$ if $\op{deg}(u)
	\ne 0$,
	\item $\op{leftSibling/rightSibling}(u)$ returns the left / right sibling of
	$u$, or $\op{null}$ if none exists.
\end{itemize}
\end{definition}
\noindent
Munro and Raman realize tree navigation by supporting the following three
operations on the balanced parenthesis representation.
\begin{lemma}\label{lem:balanced-tree} For any balanced parenthesis
	representation of length~$n$ there is an auxiliary structure on top of the
	representation that consists of~$o(n)$ bits and provides, after $O(n)$
	initialization time, the following operations in constant time: 
	\begin{itemize}
		\item $\op{findclose}(i)$ $(i \in \Nat)$: Returns the position of the
		closing parenthesis matching the open parenthesis at position~$i$.
		\item $\op{findopen}(i)$ $(i \in \Nat)$: Returns the position of the
		open parenthesis matching the closed parenthesis at position~$i$.
		\item $\op{enclose}(i)$ $(i \in \Nat)$: Returns the position of the
		closest open parenthesis part of the parenthesis pair enclosing the open
		parenthesis at position~$i$.
	\end{itemize}
\end{lemma}

The parenthesis representation is typically stored as a bit sequence $B$ with open
parenthesis represented via a $1$ and a closed parenthesis via a $0$. Via
rank-select structure (Lemma~\ref{lem:rs}) on $B$ a bidirectional mapping between each node
$u$ of an ordinal tree $T'$ and the index of $u$'s open parenthesis can be
created.

A parenthesis representation can easily be generalized to allow
not just one type of matching parenthesis, but instead use parenthesis from
a constant sized universe while providing the same operations as described
in Lemma~\ref{lem:balanced-tree} and using the same asymptotic space bound.
The techniques for this are outlined in a book on compact data structures
by Navarro~\cite[Chapter~7]{Navarro16}. Concretely, in a later section we
use a parenthesis representation that uses the alphabet $\{(,), [,]\}$. In
this section we only use the alphabet $\{(,)\}$.

To compute a balanced parenthesis representation of a tree $T$, one requires
to execute a DFS traversal of $T$. As we want our solutions to be
space-efficient, this DFS traversal must be space-efficient as well. One simple
solution is to use the space-efficient algorithm for DFS traversal of Banerjee
et al.~\cite[Lemma~2]{BanerjeeNCSRVSR18}. Their algorithm works with $O(n+m)$ bits
and time, which is $O(n)$ bits on all sparse graph classes. As described, their algorithm requires
the input graph to have so-called \textit{crosspointers}, which is a common
extension of the standard adjacency array representation. We show that one
can execute a simple linear-time and -bit DFS when only a standard adjacency array
representation of the graph is given. The general idea is quite similar to
the algorithm of Banerjee et al.~\cite{BanerjeeNCSRVSR18}. Intuitively,
we execute a standard iterator based DFS, that uses neighborhood iterators.
One can implement such an
iterator space-efficiently using self-delimiting numbers, such that the sum of
bits required is $O(n+m)$ for graphs with $n$ vertices
and $m$ edges. We summarize this
in the following lemma. 

\begin{lemma}\label{lem:simpledfs}
	Let $G$ be an $n$-vertex $m$-edge graph given via an
	adjacency array representation. There is a DFS
	that runs on $G$ in $O(n+m)$ time and bits.
\end{lemma}
\begin{proof}
    As described in Section~\ref{sec:preone}, let $\op{adj}[u]$ be the adjacency
	array of a vertex $u \in V(G)$. We begin by describing a
	small auxiliary data structure, which we call \textit{compact iterator}. It is
	initialized for a vertex $u \in V(G)$ and allows iteration over all
	neighbors of $u$ in $G$ in addition to outputting the parent $v$ of $u$ in
	the DFS tree (or $\textrm{null}$ if the vertex has no parent). The parent
	$v$ of $u$ is passed as an optional parameter during initialization. The
	exact operations that the compact iterator provides
	for a vertex $u$ are as follows:
    \begin{itemize}
		\item $\op{initialize}(u, v)$: Given a vertex $u$ and an
			 (optional) parent $v$ of $u$ in the DFS tree, initialize the
			 compact iterator by setting an iterator variable $i$ to $0$,
			 followed by iterating over $\op{adj}[u]$ to find and store the
			 index $j$ so that $v=\op{adj}[u][j]$.
        \item $\op{hasNext}()$: Output $\op{true}$ exactly if the
            iteration over $\op{adj}[u]$ is finished, i.e., $i=\op{deg}(u)$.
        \item $\op{next}()$: Output the next element $\op{adj}[u][i]$
            followed by incrementing $i$, if $\op{hasNext}()$ is
            $\op{true}$. Otherwise output $\textrm{null}$.
        \item $\op{parent}()$: Output the parent vertex $v$ of $u$ in
            the DFS tree as $\op{adj}[u][j]$, or $\op{null}$ if $u$ has no
            parent. 
    \end{itemize}

	As described, each compact iterator of a vertex $u$ consists of two
	variables $i$ and $j$, where $i$ is the current index of the iterator and
	$j$ is chosen such that $\op{adj}[u][j]$ is the parent of $u$ in the DFS
	tree. The value of each of these two variables is at most $\op{deg}(u)$ and
	so can be stored with $2 \lceil \log
	\op{deg}(u) \rceil + 1=O(\log \op{deg}(u))$ bits as a self-delimiting
	number (Section~\ref{sec:intro}). Note that the runtime for the $\op{initialize}$
	operation is $O(\op{deg}(u))$. Storing the compact iterators for all
	vertices can be done using $m$ bits in total, since $\sum_{u \in
	V(G)}O(\log \op{deg}(u)) = O(m)$. To allow constant time access to a compact
	iterator we store them contiguously for each vertex $u \in (1, 2, \ldots,
	n)$ and use static space allocation (Section~\ref{sec:preone}) to address
	them. We denote by $I$ the structure that stores all these compact iterators
	and with $I[u]$ the compact iterator for a vertex~$u$.
	
	In the following we describe our DFS algorithm. We assume that the
	structure $I$ is constructed in addition to an $n$-bit array $B$, which we
	use to mark vertices as visited/unvisited. Initially, all compact iterators
	in $I$ are not initialized, and all bits in $B$ are set to $0$, indicating
	that all vertices are unvisited.
	To execute the DFS on an arbitrary (given) start vertex $u$ we call
	$\op{explore}(u)$, which explores a connected component.
	To explore all connected components we execute $\op{exploreAll}()$.

	\begin{description}
		\item[$\op{explore}(u)$]: 
		Mark $u$ as visited ($B[u]:=1$) and initialize 
		the compact iterator of $u$ with $I[u].\op{initialize}(u, \op{null})$ 
		(note that $u$ has no parent as it becomes the root of
		the DFS tree for a connected component).
	
		While $I[u].\op{hasNext}() = \op{true}$: Obtain the vertex $v :=
		I[u].\op{next}()$. If $v$ was visited ($B[v] = 1$), continue the loop.
		Otherwise, we initialize the compact iterator of $v$ with
		$I[v].\op{initialize}(v, u)$, mark $v$ as visited ($B[v]:=1$), set
		$u:=v$ and continue the loop.
	
		If $I[u].\op{hasNext}()=\op{false}$, we have to backtrack. We get the
		parent of $u$ via $p := I[u].\op{parent}()$. If $p \neq \textrm{null}$
		we set $u:=p$ and enter the loop of the previous paragraph. If $p =
		\textrm{null}$ we have reached the root and are finished.
	 
	 	\item[$\op{exploreAll}()$]: For $u = 1$ to $n$, if $B[u] = 0$, call $\op{explore}(u)$.
	\end{description}
 
    The described algorithm executes a DFS in time $O(n+m)$ and uses $O(n + m)$ bits 
    to store $I$ and $B$. 
    Each compact iterator is initialized
	only once, which takes a total of $O(n + m)$ time. The space used by other
	variables is negligible.
\end{proof}

As mentioned, to construct the balanced parenthesis structure that represents a tree
one simply requires a DFS. This is summarized in the following lemma. In all subsequent
results that require a tree (forest) as an input we require either the input tree (forest)
to be given via adjacency arrays (Section~\ref{sec:preone}), via a balanced parenthesis
representation or given implicitly via access to a DFS that outputs the nodes of the tree.

\begin{lemma}\label{lem:constructtree} Given a DFS that outputs the nodes of an
$n$-node tree $T$ in constant time per node there is an algorithm that computes
a data structure representing an ordinal tree $T'$ in $O(n)$ time using $O(n)$
bits such that $T$ and $T'$ are isomorphic. The data structure allows tree
navigation on $T'$ in constant time.
\end{lemma}

When working with node-colored trees $T$ one can extend the data
structure in a simple fashion to also maintain the colors on the tree $T'$
constructed via Lemma~\ref{lem:constructtree}.
We assume for the $n$-node tree $T$ that there exists a function $\rm{color}: V(T) \rightarrow \{0, 1\}^{+}$
that outputs for each node $u \in V(T)$ a binary string $\rm{color}(u)$ that represents
the color of $u$. After the construction of $T'$ we want to maintain the color
information, but as $T'$ is only isomorphic to $T$ this is not directly possible.
Roughly speaking, $T'$ is a re-labeling of $T$, and we thus need to account for this.
Note that $T'$ is labeled using a DFS on $T$, meaning each node of $T'$ is labeled
in pre-order. During the traversal, when visiting a node $u$ of $T$ for the first time,
we simply write the binary string $\rm{color}(u)$ into a bit vector $C$,
and in a secondary bit vector $R$ mark the first position in $C$ of the newly written
binary string $\rm{color}(u)$ into $C$. Once $T'$ is constructed in this fashion,
$C$ is a concatenation of all binary strings $\rm{color}(u)$ for $u \in V(T)$ sorted via pre-order traversal.
Using a rank-select structure on $R$ then allows to evaluate the colors
for $T'$, i.e., we store the colors in $C$ via static space allocation.
We summarize this in the following corollary.

\begin{corollary}\label{cor:constructcoloredtree}
	Given a rooted $n$-node colored tree $T$ and a color function $\rm{color}: V(T) \rightarrow \{0, 1\}^{+}$
	that can be evaluated in time $O(|\rm{color}(u)|)$ for each $u \in V(T)$
	there is an algorithm that computes a data
	structure representing an ordinal colored tree $T'$ in $O(n+b)$ time using $O(n+b)$ bits
	such that $T$ and $T'$ are isomorphic, with $b=\sum_{u \in V(T)} |\rm{color}(u)|$.
	The data structure allows tree navigation
	on $T'$ in constant time in addition to access to the $\rm{color}$ function for each $u' \in V(T')$.
\end{corollary}

\begin{figure}[h]
\centering
	\begin{subfigure}{.48\textwidth}
		\centering
		\begin{tikzpicture}[scale=0.8,every tree node/.style={minimum width=1.28em, minimum height=1.28em, draw, circle, inner sep=0pt, text width=10pt, align=center, fill=white, font=\footnotesize},
level distance=2em,sibling distance=0.374em, 
edge from parent path={(\tikzparentnode) -- (\tikzchildnode)}]
\Tree [.\node (a) {9}; 
[.\node[] (a) {11};
	[.\node[] (a) {10}; ] [.\node[] (a) {3}; ] [.\node[] (a) {13}; ] [.\node[] (a) {5}; ] 
] [.\node[] (a) {12}; ] [.\node[] (a) {4}; ] [.\node[] (a) {7}; 
	[.\node[] (a) {14}; ] [.\node[] (a) {2}; ] [.\node[] (a) {6}; ] [.\node[] (a) {8}; ] 
]  [.\node[] (a) {1}; ]
]
\end{tikzpicture}%
		\caption{A rooted tree~$T$.}	
	\end{subfigure}
	\begin{subfigure}{.48\textwidth}
		\centering
		\begin{tikzpicture}[scale=0.8,every tree node/.style={minimum width=1.28em, minimum height=1.28em, draw, circle, inner sep=0pt, text width=10pt, align=center, fill=white, font=\footnotesize},
level distance=2em,sibling distance=0.374em, 
edge from parent path={(\tikzparentnode) -- (\tikzchildnode)}]
\Tree [.\node (a) {1}; 
[.\node[] (a) {2};
	[.\node[] (a) {3}; ] [.\node[] (a) {4}; ] [.\node[] (a) {5}; ] [.\node[] (a) {6}; ] 
] [.\node[] (a) {7}; ] [.\node[] (a) {8}; ] [.\node[] (a) {9}; 
	[.\node[] (a) {10}; ] [.\node[] (a) {11}; ] [.\node[] (a) {12}; ] [.\node[] (a) {13}; ] 
]  [.\node[] (a) {14}; ]
]
\end{tikzpicture}
		\caption{An ordinal tree $T'$ isomorphic to~$T$.}
	\end{subfigure}
	\vspace{10pt}
	\begin{tikzpicture}[
every node/.style={minimum width=1.18em, minimum height=1.18em, inner sep=0pt, text width=8pt, align=center, fill=white, font=\footnotesize},
start chain = going right,
node distance = 0pt]
\field{1}{\bfseries  (}{1}
\field{2}{\bfseries  (}{2}
\field{3}{\bfseries  (}{3} 
\field{4}{)}{} 
\field{5}{\bfseries  (}{4}
\field{6}{)}{} 
\field{7}{\bfseries  (}{5}
\field{8}{)}{}
\field{9}{\bfseries  (}{6} 
\field{10}{)}{}
\field{11}{)}{}
\field{12}{\bfseries  (}{7} 
\field{13}{)}{}
\field{14}{\bfseries  (}{8} 
\field{15}{)}{}
\field{16}{\bfseries  (}{9}
\field{17}{\bfseries  (}{10} 
\field{18}{)}{} 
\field{19}{\bfseries  (}{11}
\field{20}{)}{} 
\field{21}{\bfseries  (}{12}
\field{22}{)}{}
\field{23}{\bfseries  (}{13} 
\field{24}{)}{}
\field{25}{)}{}
\field{26}{\bfseries  (}{14} 
\field{27}{)}{}
\field{28}{)}{} 

\connectBelow[0.5]{1}{28}

\connectBelow[0.4]{2}{11}
\connectBelow[0.4]{16}{25}

\connectBelow[0.3]{3}{4}
\connectBelow[0.3]{5}{6}
\connectBelow[0.3]{7}{8}
\connectBelow[0.3]{9}{10}

\connectBelow[0.3]{12}{13}
\connectBelow[0.3]{14}{15}

\connectBelow[0.3]{17}{18}
\connectBelow[0.3]{19}{20}
\connectBelow[0.3]{21}{22}
\connectBelow[0.3]{23}{24}

\connectBelow[0.3]{26}{27}
\end{tikzpicture}
	\caption{Given a rooted tree $T$ (shown in (a)) with the root $u = 9$ a
	balanced parenthesis representation $\op{bpr}_u$ (below of (a) and of (b))
	of an ordinal tree $T'$ that is isomorphic to $T$ (shown in (b)) can be
	computed by a single DFS run on $T$ in preorder. The parenthesis pairs of
	each node of $T'$ are connected by a line and the index of each open
	parenthesis represents the label of that node in $T'$.}
	\label{fig:ordinal-tree}
\end{figure}

\section{Tree Isomorphism}\label{sec:treeiso}%
We start this section by giving an introduction to graph and tree isomorphism.

\begin{definition}{(graph isomorphism)}\label{def:graph-iso} Two given graphs $G
	= (V_G, E_G)$ and $H = (V_H, E_H)$ are isomorphic exactly if there exists a
	bijective mapping $f: V_G \rightarrow V_H$ such that vertices $u, v \in V_G$
	are adjacent in $G$ if and only if $f(u), f(v) \in V_H$ are adjacent in $H$.
	Then~$f$ is called an {\em isomorphic mapping}.
\end{definition}

Buss~\cite{Buss97} shows an equivalent definition for rooted trees.
\begin{definition}{(rooted tree isomorphism)}\label{def:treeiso} By induction
	two rooted trees $T$ and $T'$ are isomorphic if and only if
	\begin{enumerate}[label=(\alph*)]
		\item $T$ and $T'$ consist of only one node, or
		\item the roots $r$ and $r'$ of $T$ and $T'$, respectively, have the
		same number $m$ of children, and there is some ordering $T_1, \ldots,
		T_m$ of the maximal subtrees below the children of $r$ and some ordering
		$T'_1, \ldots, T'_m$ of the maximal subtrees below the children of $r'$
		such that $T_i$ and $T'_i$ are isomorphic for all $1 \le i \le m$.
	\end{enumerate}
\end{definition}

We start to describe a folklore algorithm for tree isomorphism that
requires $\Theta(n \log n)$ bits on $n$-node trees. Let $T_1$ and $T_2$ be two
rooted trees. The algorithm processes the nodes of each tree in rounds. In each
round, all nodes of depth $d = \text{max}, \ldots, 0$ are processed. Within a
round, the goal is to compute a {\em classification number} for every node~$u$
of depth~$d$, i.e., a number out of $\{0, \ldots, n\}$ that represents the
structure of the maximal subtree below $u$. The correctness of the algorithm is
shown in~\cite{AhoHU74} and follows from the invariant that two subtrees in the
trees $T_1$ and $T_2$ get the same classification number exactly if they are
isomorphic.

Since we later want to modify the algorithm, we now describe it in greater
detail. In an initial process assign the classification number~$0$ to every leaf
in each tree. {Then, starting with the maximal depth do the following: start by
computing (for each tree) the {\em classification vector} of each non-leaf $v$
of depth $d$ consisting of the classification numbers of $v$'s children, sorted
lexicographically. After doing this in each tree, compute the classification
number for the non-leafs as follows: for each tree $T_1, T_2$ put the
classification vectors of depth $d$ into a single sequence $\mathcal{S}_1$ and
$\mathcal{S}_2$, respectively. Sort each of these sequences lexicographically by
interpreting each classification vector in the sequence as a number. Then
assign classification numbers $1, 2, 3,$ etc.\ to the vectors in the (now
sorted) sequences $\mathcal{S}_1, \mathcal{S}_2$ such that two vectors get the
same number exactly if they are equal (among both sequences).} By induction the
invariant holds for all new classification numbers. Repeat the whole procedure
iteratively for the remaining depths until reaching the root. By the invariant
above, both trees are isomorphic exactly if the roots of both trees have the
same classification number.

\begin{figure}[h!]
	\centering
	\begin{subfigure}{0.5\textwidth}
		\centering
		\begin{tikzpicture}[scale=0.9,every tree node/.style={minimum width=1.28em, minimum height=1.28em, draw, circle, inner sep=0pt, text width=8pt, align=center, fill=white, font=\footnotesize},
level distance=2em,sibling distance=0.374em, 
edge from parent path={(\tikzparentnode) -- (\tikzchildnode)}]
\Tree [.\node (n1) {1}; 
	[.\node (3) {1};
		[.\node (n2) {0};]
		[.\node (2) {1};
			[.\node (1) {0};]
		]
	][.\node (n3) {2};
		[.\node (n4) {0};]
		[.\node (n5) {0};]
		[.\node (n6) {0};]
	]
	[.\node (n7) {0};]
]

\node[right=6.5em of 1] (v1) {\small $(({0}))$};
\node[right=6.5em of 2] (v2) {\small$(({0},{1})({0},{0},{0}))$};
\node[right=7.2em of 3] (v3) {\small$(({0},{1},{2}))$};

\node[above=-0.6em of v1] () {1};
\node[above=-0.6em of v2, shift={(-0.6,0)}] () {1};
\node[above=-0.6em of v2, shift={(0.42,0)}] () {2};
\node[above=-0.6em of v3] () {1};

\end{tikzpicture}
	\end{subfigure}%
	\begin{subfigure}{0.5\textwidth}
		\centering
		\begin{tikzpicture}[scale=0.9,every tree node/.style={minimum width=1.28em, minimum height=1.28em, draw, circle, inner sep=0pt, text width=8pt, align=center, fill=white, font=\footnotesize},
level distance=2em,sibling distance=0.374em, 
edge from parent path={(\tikzparentnode) -- (\tikzchildnode)}]
\Tree [.\node (n1) {1};
	[.\node (n2) {2};
		[.\node (n3) {0};][.\node (n4) {0};][.\node (n5) {0};]
	]
	[.\node (n6) {0};]
	[.\node (3) {1};
		[.\node (2) {1};
			[.\node (1) {0};]
		]
		[.\node (n7) {0};]
	]
]

\node[right=2em of 1] (v1) {\small $(({0}))$};
\node[right=2em of 2] (v2) {\small$(({0},{1})({0},{0},{0}))$};
\node[right=1.3em of 3] (v3) {\small$(({0},{1},{2}))$};

\node[above=-0.6em of v1] () {1};
\node[above=-0.6em of v2, shift={(-0.6,0)}] () {1};
\node[above=-0.6em of v2, shift={(0.42,0)}] () {2};
\node[above=-0.6em of v3] () {1};
\end{tikzpicture}
	\end{subfigure}
	\caption{An example of the algorithm with                  
                  isomorphic trees $T_1$ (left) and $T_2$ (right). The numbers
	inside the nodes and above the vectors are classification numbers. The
	vectors within the parenthesis are the sorted classification vectors of
	maximal subtrees.}\label{fig:folklore-tree-iso}
\end{figure}

The algorithm above traverses the nodes in order of their depth, starting from
the largest and moving to the smallest, until it reaches the root. One key
modification we make to achieve our goal of making the aforementioned algorithm
space efficient, is that we traverse the nodes in order of their height starting
from height $0$ (first round) until reaching the root with the largest height
(last round), i.e., in \textit{increasing height}. As mentioned, the standard
algorithm requires the nodes of the tree to be output in {\em shrinking depth}.
While there is a succinct data structure due to He et
al.~\cite{he_et_al:LIPIcs:2020:13369} that provides us with all necessary
operations to implement such a shrinking depth tree traversal, the construction
step of this data structure requires $O(n \log n)$ bits due to (among other
things) the usage of an algorithm of Farzan and
Munro~\cite[Theorem~1]{DBLP:journals/algorithmica/FarzanM14} that partitions the
input tree into smaller subtrees. The aforementioned algorithm uses a stack of
size $O(n)$, with each element on the stack using $\Theta(\log n)$ bits. In
addition, the decomposition itself is stored temporarily. As we aim for a space
usage of $O(n)$ bits we use a different approach. Note that the tree data structure
outlined in Section~\ref{sec:prelim} only allows to output the nodes in order of
their \textit{increasing depth}. While this implies a simple algorithm for
shrinking depth, by storing the result of the traversal and then outputting it
in reverse order, such a solution requires $\Theta(n \log n)$ bits. We therefore
modify the standard isomorphism algorithm such that traversal by
\textit{increasing height} suffices, which can be implemented using the tree
navigation provided by the data structure outlined in Section~\ref{sec:prelim}.

The difference to the standard approach is that we compute classification vectors
consisting of classification numbers, which were computed in different rounds.
To avoid a non-injective mapping of the subtrees, our classification numbers
consist of tuples $(h_u, q_u)$ for each node $u$ where $h_u$ is the height of
$u$ and $q_u$ is a number representing the subtree induced by $u$ and its
descendants. Intuitively, $q_u$ is the classification number from the folklore
algorithm above.

The same invariant as for the standard algorithm easily shows the correctness of
our modified algorithm whose space consumption 
is determined by
\begin{enumerate}[label=(\Alph*)]
	\item the space for traversing nodes in order of their height,
	\item the space for storing the classification vectors and the
              classification numbers, and
	\item the space needed by an algorithm to assign new classification numbers
	based on the previous computed classification vectors.
\end{enumerate}
We now describe $O(n)$-bit solutions for (A)-(C).

{\bfseries (A) Iterator returning vertices in increasing height.} The idea of
the first iteration round is to determine all nodes of height $h = 0$ (i.e., all
leaves) of the given tree in linear time and to store them in an $O(n)$-bit
choice dictionary~$C$ (Lemma~\ref{lem:cd}). While iterating over the nodes of height $h$ in $C$, our
goal is to determine the nodes of height $h + 1$ and store them in a choice
dictionary $C'$. The details are described in the next paragraph. If the
iteration of the nodes in $C$ is finished, swap the meaning of $C$ and $C'$ and
repeat this process iteratively with $h + 1$ as the new height $h$ until,
finally, the root is reached.

A node is selected for $C'$ at the moment when we have processed all of its
children. To compute this information, our idea is to give every unprocessed
node $u$ a token that is initially positioned at its leftmost child. Intuitively, the goal is
to pass that token over every child of $u$ from left to right until reaching the
rightmost child of $u$, at which point we mark $u$ as processed and store it in
$C'$. More precisely, we iterate over the
children in the deterministic order given via the iterate operation provided 
by the choice dictionary.
Initially, no node is marked as processed. Informally speaking, we run a
relay race where the children of a node are the runners. Before runners can
start their run, they must be marked as processed. The initiative to pass the
token is driven by the children of $u$. Whenever a child $v$ of $u$ is
processed, we check if either $v$ is the leftmost child of $u$ or $v$'s left
sibling has $u$'s token. If so, we move the token to the right sibling $v'$ of
$v$ and then to the right sibling of $v'$, etc., as long as the sibling is
already marked. If all children of~$u$ are processed, $u$ becomes marked and
part of $C'$.

Using Lemma~\ref{lem:balanced-tree} we can jump from a node to its previous and
next sibling in constant time. Thus, all moves of $u$'s token can be processed
in $O(\op{deg}(u))$ time until we insert $u$ in $C'$. In total, moving all
tokens can be done in $O(n)$ time. Based on the ideas above, we get the
following lemma.

\begin{lemma}\label{lem:iterHeigth} Given an unlabeled rooted $n$-node tree $T$
	there is an iteration over all nodes of the tree in order of their height
	that runs in linear time and uses $O(n)$ bits. The iteration is realized by
	the following methods:
\begin{itemize}
	\item \op{init}$(T)$: Initializes the iterator and sets height $h = 0$.
	\item \op{hasNext}: Returns \textsc{true} exactly if nodes of height $h$
	exist.
	\item \op{next}: Returns a choice dictionary containing nodes of height~$h$
	and increments $h$ by 1.
\end{itemize}
\end{lemma}
\begin{proof}
Let $T=(V, E)$ be the tree under consideration. We first describe the
initialization. Compute a balanced parenthesis representation $R$ of $T$ such
that tree navigation is supported~(Lemma~\ref{lem:balanced-tree}). Let $u_($ and
$u_)$ be the positions of the open and closed parenthesis, respectively, for
each node $u$ in $R$ and, let $n' = 0$ be the number of nodes already returned
by the iterator. Initialize $C$ and $C'$ as empty choice dictionaries of $O(n)$
bits. Take $D$ as a bit-vector of $2n$ bits where we mark a node $u$ as
processed by setting $D[u_(] = 1$ and take $P$ as a bit-vector of $2n$ bits
where we store all tokens for the nodes. Intuitively speaking, $P$ is used as a
raw area of memory (initially containing only $0$-bits) which we use to store
the tokens. In detail the token for node $u$ is stored in $P[v^1_(], \ldots,
P[v^{\ell}_(]$ where $v^1, \ldots, v^{\ell}$ are the children of $u$. In the
subsequent steps we adhere to the following two invariants before and after each
round:
\begin{enumerate}
	\item For a child $u$ of a node $v$, $P[u_(] = 1$ holds exactly if, for all
siblings $w$ left to $u$, $D[w_(] = 1$ holds and for the next right sibling $y$
of $u$, $D[y_(] = 0$ holds.
	\item If $D[u_(] = 1$ for all children $u$ of a node $v$, then $D[v_(] = 1$.
\end{enumerate}
For any call of \op{hasNext}, return $(n' < n)$. {If $\op{next}$ is called with
$n' > n$, return a \op{null} pointer.} For a call of \op{next} with $n' = 0$, do the following:
traverse $R$ to find all leaves, i.e., for each node $u$ with $u_) = u_( + 1$,
store $u$ into $C$, set $D[u_(] = 1$, and set $n'$ to the number of leaves.
Afterwards return a pointer to $C$.

For a call of \op{next} with $n' > 0$ {(and $n' \leq n$)}, collect the nodes of
the next height as follows:
iterate over the nodes $u$ in $C$ and proceed to check if $u$ is a leftmost
 child or check for the left sibling $v$ of $u$ if $D[v_(] = 1 \land P[v_(] = 1$
 holds. If not, continue the iteration over $C$. Otherwise, start or continue
 the {``relay race'': iteratively move to the right sibling(s) $v$ of $u$ as
 long as $D[v_(] = 1$ holds. Then, set $P[v_(] = 1$. Now, if $v$ is the
 rightmost sibling (i.e., $v$ has no right sibling), determine the parent $x$ of
 $u$, set $D[x_(] = 1$, add $x$ to $C'$ and increment $n'$ by one. After
 executing this procedure for all nodes in $C$,} set $C := C'$, initialize a new
 empty choice dictionary~$C'$ and return a pointer to $C$.

Note that all our structures use $O(n)$ bits and can easily
be initialized in $O(n)$ time. {Each operation regarding the
balanced parenthesis representation $R$ takes constant time each (e.g., tree
navigation, finding matching opening/closing parenthesis, \ldots).} Moreover, we
pass the token only over siblings that are already marked in $D$ and do nothing
if we cannot pass a token to the next sibling. In total, we require time linear
to the amount of siblings. To conclude, the iteration over all nodes of the
given tree requires linear time. 
\end{proof}

{\bfseries (B) Storing the classification numbers.} We now describe an
algorithm to store our classification numbers in an $O(n)$-bit storage schema.
Recall that a classification vector of a node $u$ consists of the classification
numbers of its children. Our idea is to use self-delimiting numbers to store the
classification numbers and to choose the classification numbers such that their
size is bounded by \op{max}($c_1 \cdot \log n$, $c_2 \cdot \op{desc}(u)$) for some
constants $c_1, c_2$.
We take $(0, 0)$ as a classification number for every leaf so that a constant
number of bits suffices to represent it. Observe that, after computing the
classification number of a node~$u$ the  
classification vector of $u$ (i.e., the classification numbers of $u$'s
children) is not required anymore and thus, 
the classification vector 
(and thus, the classification numbers of $u$'s children) can be overwritten.
We use this
space to store the classification number of $u$.
By the next lemma we can store the classification numbers and vectors.
(The function $\rm{color}$, required by the lemma, will be used later 
to adjust the space needed to also encode colors for node-colored trees.
For uncolored trees $\rm{color}(\cdot)$ returns an empty bit string for all inputs.)

\begin{lemma}\label{lem:dataSchema} 
     Let $c > 0$ be 
     a constant integer and $T$ 
     a node-colored rooted $n$-node tree with $\rm{color}: V(T) \rightarrow \Nat$
      the function to evaluate the color of each node $u \in V(T)$ such that
	  $\rm{color}(u)$ can be evaluated in time $O(|\rm{color}(u)|)$. Then there is an
	$O(n + \sum_{u \in V(T)}|\rm{color}(u)|)$-bits data structure that
	initializes in $O(n + \sum_{u \in V(T)}|\rm{color}(u)|)$ time and, for each node $u$ of $T$, provides
	operations \op{read}($u$) and \op{write}($u$) in constant time as
	well as
	\op{vector}($u$) in $O(\rm{deg}(u))$ time.
	\begin{itemize}
		\item \op{read}($u$) ($u$ node of $T$): If a number $x$ is stored for
                    $u$, then $x$ is returned. Otherwise, { the result is
                    undefined}.
		\item \op{write}($u$, $x$) ($u$ node of $T$, $0 \le x \le
            \min\{2^{2c\cdot \op{desc(u)} + |\rm{color}(u)|},\mathrm{poly}(w)\}$): Store number $x$ for
            node $u$ and delete all stored numbers of the descendants of~$u$.
		\item \op{vector}($u$) ($u$ node of $T$): Returns the bit-vector of
		length $\le 2c\cdot \op{desc}(u) + \sum_{v \in \rm{children}(u)} |\rm{color}(u)|$ consisting of the concatenation of the
		self-delimiting numbers stored for the children of $u$.
	\end{itemize}
\end{lemma}
\begin{proof}
	We assume that $T$ is given via the data structure of
	Corollary~\ref{cor:constructcoloredtree}, otherwise
	construct it.
	We first describe some auxiliary data structures that we require.
	The first is a bit array $B$ of $6c\cdot n + \sum_{u \in V(T)}|\rm{color}(u)|$ bits 
	that we use to maintain the numbers, i.e., the values we read/write
	with the specified operations.
	We use a second bit array $P$ under $B$ where we mark the position in $B$
	for each node $u$ of $T$ where we can store the given number for $u$.
    (Note that, if $|\rm{color}(u)| = 0$ for all $u$ in $T$ we can calculate the exact positions without $P$
	and write into $B[6c\cdot u_(, 6c\cdot u_)]$, by exploiting the fact that for each node $u$, 
	$u_($ and $u_)$ are the positions of the open and closed parenthesis of a node $u$, respectively.)
	Given the balanced parenthesis representation we can iterate over the nodes of $T$
    in pre-order and flip the $i$th bit in $P$ to $1$ where $i$ is initially $0$ and is incremented after
	each iteration over a node $u$ by $6c\cdot \rm{desc}(u) + |\rm{color}(u)|$.
	Since by construction the $i$th flipped bit corresponds to the value of $u$,
	we can find the position $p$ of the $i$th bit by calculating a 
	rank-select data structure on $P$ and call $p := \rm{select}(u)$,
    i.e., we use static space allocation.
	Having $p$ we realize \op{write}($u$, $x$) and by transforming $x$ into a self-delimiting number
	and write it starting at $B[i]$.
	The operation \op{read}($u$) is realized by reading a self-delimiting number starting from $B[i]$
	and transforming it into a number.
	To return \op{vector}($u$), call \op{read}($v$) for all children of $u$
    and concatenate the returned self-delimiting numbers to one bit-vector.
    
    We now argue that the bits required to store a number $x$, written for a node $u$,
    fits into $B[p_u, p_v)$ where $v$ is the next sibling of $u$ and $p_u = \rm{select}(u)$
    and $p_v = \rm{select}(v)$.
	Recall that a number $x$ can be stored as a self-delimiting
	number with $f(x) = 3$ bits if $x \le 1$ and with $f(x) = 2\lceil \log x
	\rceil + 1 \le 3\lceil \log x \rceil$ bits if $x > 1$. 
    Since $x \le \min\{2^{2c\op{desc(u)} + |\rm{color}(u)|},\mathrm{poly}(w)\}$, $x$ can be stored
    inside a space of $6c\cdot\op{desc}(u) + |\rm{color}(u)|$ bits.
	By construction of $P$ this is exactly the space between $p_u$ and $p_v$.
	
	Observe that by the pre-order traversal using the balanced parenthesis representation
	the positions of all children of $u$ are between $p_u$ and $p_v$ (if a sibling $v$ exists),
	hence writing $x$ can overwrite only the numbers of the children of $u$.
	
	Concerning the space note that the balanced parenthesis representation uses $O(n)$ bits.
	The auxiliary structures $B$ and $P$ use $O(n + \sum_{u \in V(T)}|\rm{color}(u)|)$ bits
	and a rank-select structure on $P$ can be constructed with the same space bound.
	Hence, our space bound is as stated in the lemma.
	
	The initialization of $B$, $P$ and the rank-select structure requires $O(n + \sum_{u \in V(T)}|\rm{color}(u)|)$ time,
	since it is determined by the space of $B$ and $P$. The balanced parenthesis representation can be computed
	within the same time.
    Writing and reading a number $x$ stored with $O(\log x)$ bits requires $O(\log x / w) = O(1)$ time
	on the word RAM. Reading $\rm{vector}(u)$ consequently requires $O(\rm{deg}(u))$ time. 
%
%
\end{proof}

{\bfseries (C) Computing classification numbers.} Let $D$ be our data structure
of Lemma~\ref{lem:dataSchema} where we store all classification vectors. Our
next goal is to replace the classification vector $D.\op{vector}(u)$ of all
processed subtrees with root $u$ and height $h$ by a classification number $(h,
q_u)$ with $|q_u| = O(\min\{\op{desc}(u),\log n\})$ such that the
componentwise-sorted classification vectors are equal exactly if they get the
same classification number.

Our idea is to sort $D.\op{vector}(u)$ by using Theorem~\ref{thm:sdn-sorting} to
obtain a component-wise sorted classification vector and turn this vector into a
self-delimiting number for further operation on it. We subsequently compute the
dense rank to replace the self-delimiting number in $D.\op{vector}(u)$ by the
tuple (height, dense rank).
To make it work we transform each vector into a self-delimiting number by
considering the bit-sequence of the vector as a number (i.e., assign the prefix
$0^\ell0$ to each vector where $\ell$ is the length of the vector in bits). We
can store all these vectors as self-delimiting numbers in a bit-vector $Z_{h +
1}$ of $O(n)$ bits. Then we can use Theorem~\ref{thm:dense-rank} applied to
$Z_{h + 1}$ to compute the dense ranks, which allows us to determine the
classification numbers for all subtrees of height $h+1$.\\

{\bfseries Our Algorithm.} We now combine the solutions for (A)-(C).

\begin{lemma}\label{th:treeiso} Given two rooted $n$-node trees $T_1$ and $T_2$
	there is an algorithm that recognizes if $T_1$ and $T_2$ are isomorphic in
	linear-time using $O(n)$ bits.
\end{lemma}
\begin{proof}
Computing the leaves and their classification numbers is trivial. So assume we
have computed the classification numbers for the subtrees of height $h$, and we
have to compute the classification numbers for the subtrees of height $h + 1$.
We iterate over every subtree with root $u$ and height $h + 1$ by using the
choice dictionary $\mathcal C$ returned by our iterator of
Lemma~\ref{lem:iterHeigth}.

We assume a data structure $D$ of Lemma~\ref{lem:dataSchema} is constructed.
In the following, let $k_u$ be the number of children of a node $u$, let $\bb_u$
be the number of bits of $D.\op{vector}(u)$ and let $V_i$ be the set of nodes of
height~$i$ in both trees. For each node $u$, we sort the classification numbers
of its children stored in $D.\op{vector}(u)$ into a new bit-vector $X$ of
$\bb_u$ bits by applying Theorem~\ref{thm:sdn-sorting} (with parameters $\bb =
\bb_u$, $k = k_u$ and $1\log n$) and store it by calling
$D.\op{write}(u, X)$.

We can now construct the vectors $Z^{T_1}_{h + 1}$ and $Z^{T_2}_{h + 1}$ for the
subtrees of height $h+1$ within the trees $T_1$ and $T_2$, respectively.
$Z^{T_j}_{h + 1}$ consists of the self-delimiting numbers constructed from the
vector $D.\op{vector}(u)$ of each node $u$ in $T_j$ of height $h + 1$,
with $j \in \{1, 2\}$. We then
put $Z^{T_1}_{h + 1}$ and $Z^{T_2}_{h + 1}$ together into a total sequence $Z$
and apply Theorem~\ref{thm:dense-rank} (with parameters $\bb = O(\sum_{u \in
V_{h + 1}} \bb_u)$, $k$ being the number of self-delimiting numbers in $Z$ and
$\tau=\log n$) to determine the dense rank of the numbers in $Z$.
Subsequently, we iterate over the numbers $q$ in $Z$ and use their position $p$
to retrieve their {dense} rank $r_q = R(p, q)$. In parallel, we iterate over the
subtrees of height $h + 1$ in both trees again and replace the classification
vector~$q$ by calling $D.\op{write}(u, (h + 1, r_q))$.
After computing the classification number for each node of height $h + 1$ in
$T_1$ and $T_2$, we can compare them and proceed with the next height until
reaching the root in both trees. Following this, we can easily answer the
question if $T_1$ and $T_2$ are isomorphic or not.

One can easily see that all substeps run with $O(n)$ bits. Let us next consider
the parameters of the calls to Theorem~\ref{thm:sdn-sorting}. As mentioned, we
use $\tau=\log n$ for all calls for Theorem~\ref{thm:sdn-sorting} and
Theorem~\ref{thm:dense-rank}. Parameter $k$ summed over all calls is bounded by
$\sum_u \mathrm{deg}(u)=O(n)$ and parameter $\bb$ over all calls is bounded by
$\sum_u \mathrm{deg}(u) \log n=O(n\log n)$ since we made sure that the
classification number can be stored in $O(\log n)$ bits (a classification number
is made up of tuples with values of size $O(n)$). Thus, the total time is
bounded by $O(k+\bb/\tau)=O(n)$.

We next consider the calls to Theorem~\ref{thm:dense-rank}. Since every {node}
is part of one set $V_h$, $k$ summed over all calls is bounded as before by
$O(n)$. Moreover, each {node} with its classification number is part of one
classification vector. Thus, $\bb$ over all calls is bounded as before by
$O(n\log n)$. Due to our choice of $\tau$, we obtain a total running time of
$O(n)$. 
\end{proof}

We generalize Lemma~\ref{th:treeiso} to unrooted trees by first determining the
center of a tree space efficiently, which is a set of constant size that is
maximal distance to a leaf. Using each vertex in the center as a (possible)
root, we then run our rooted isomorphism algorithm. A {\em center} of a tree
$T$ is a set consisting of either one node or two adjacent nodes of~$T$ such
that the distance to a leaf is maximal for these nodes compared to all other
nodes of~$T$. It is known that every tree has a unique
center~\cite[Theorem~4.2]{Harary91}. If two trees $T_1$ and $T_2$ are
isomorphic, then every isomorphism function maps the center nodes of~$T_1$ to
the center nodes~of~$T_2$.

\begin{lemma}\label{lem:select-root} Given an $n$-node tree $T=(V, E)$ there is
	a linear-time $O(n)$-bit algorithm that determines the center of $T$.
\end{lemma}
\begin{proof}
	Harary~\cite[Theorem~4.2]{Harary91} determines the center of a tree in
	rounds where, in each round, all nodes of degree~$1$ are removed. The
	algorithm stops if at most two nodes are left, which are selected as the
	center.

    We follow a similar approach, but cannot afford to manipulate the given
    trees. For our approach we need to be able to read and reduce the degree of
    a node. We also can not use Lemma~\ref{lem:iterHeigth} since it requires a
    rooted tree. Alternatively, we store the initial degrees of
        all nodes as self-delimiting numbers by using static space allocation
    (Section~\ref{sec:preone}).
	
	Similar to Harary's algorithm, our algorithm also works in rounds. We use a
	choice dictionary $C$ consisting of the nodes of degree one, which we
	initially fill by scanning through the whole tree. Then, while we iterate
	over each node $u = C.\op{choice}$, we delete $u$ from $C$, increment an
	initial-zero global counter $k$ to be able to determine how many nodes
	remain in the tree, and decrement the degree of the only node $v$ (neighbor
	of $u$) with $\op{deg}(v) \ge 2$. Intuitively speaking this simulates the
	deletion of $u$. If the degree of $v$ is 1 after decrementing, we add $v$ in
	a second choice dictionary $C'$. After the iteration over the nodes in $C$
    is finished, we swap the meaning of $C$ with $C'$. If at most two nodes remain (i.e, $k \ge
	n - 2$), we output the nodes in $C$ as the center of the tree. Otherwise, we
	start the next round.
	 
	Storing two choice dictionaries and all degrees can be done with $O(n)$
	bits. Moreover, since the total degree in a tree is bounded by $O(n)$, our
	approach runs in $O(n)$ time. 
\end{proof}

We can check isomorphism for two given non-rooted trees $T_1$ and $T_2$ by first
computing the center $C_{T_1}$ and $C_{T_2}$ of $T_1$ and $T_2$, respectively.
Then, we iterate over the $4$ possibilities to choose a root $r_{T_1} \in
C_{T_1}$ and $r_{T_2} \in C_{T_2}$ and apply Lemma~\ref{th:treeiso}. $T_1$ and
$T_2$ are isomorphic exactly if one possible rooted versions of $T_1$ and $T_2$
are isomorphic. We thus can conclude the next theorem.

\begin{theorem}\label{th:unrooted-treeiso} Given two (unrooted) trees $T_1$ and
	$T_2$ there  is an algorithm that outputs if $T_1$ and $T_2$ are isomorphic
	in linear-time using $O(n)$ bits.
\end{theorem}
We finally consider isomorphism on node-colored trees.
We view node-colored trees as trees where
each node has a color assigned in the form of a bit string.
Two colored trees $T_1 = (V_{T_1}, E_{T_1})$ and $T = (V_{T_2}, E_{T_2})$ are
isomorphic if an isomorphic mapping $f: V_{T_1} \rightarrow V_{T_2}$ exists such
that for each $u \in V_{T_1}$, $u$ and $v = f(u)$ have the same color assigned.

To check isomorphism of uncolored trees we use a classification number whose
computation was only influenced by the structure of the tree. To check
isomorphism of node-colored trees we additionally use the color of a node to compute
a classification number, i.e., we add to each classification vector of a node
$u$ the color of $u$ as the last entry of the vector.


\begin{theorem}\label{lm:colored-trees} Given two (unrooted) colored $n$-node
	trees $T_1$ and $T_2$, there is an algorithm that outputs if $T_1$ and $T_2$
	are isomorphic in $O(n + b)$ time using $O(n + b)$ bits exactly if the colors of
	all nodes of $T_1$ and $T_2$ can be written within~$O(b)$ bits.
\end{theorem}
\begin{proof}

We assume w.l.o.g. that both trees are rooted, otherwise we do the same as for
uncolored tree, i.e., find the centers of both trees with Lemma~\ref{lem:select-root}
and repeat this theorem with the 4 possibilities to choose the correct roots.
We start by describing a pre-processing
step. We describe this for an $n$-node tree $T$. 
The process is then applied to $T_1$ and $T_2$, respectively.
Let $\rm{color}(u) \in \{0,1\}^*$ be the color of a node $u$ in $T$.
First we transform $T$ into a 
colored ordinal tree (by using Corollary~\ref{cor:constructcoloredtree})
represented by a balanced parenthesis representation. 

Our general approach to check isomorphism remains the same, but we have to make sure
that two nodes with non-isomorphic maximal subtrees cannot get the 
same classification number. We do this by adding the color of a node
to the classification number, which changes the classification number from
a tuple $(h_u, q_u)$ to a triple $(h_u, q_u, \rm{color}(u))$,
where $h_u$ is the height of $u$ in $T$, $q_u$ is the number representing
the structure of the maximal induced subtree and $\rm{color}(u)$ is (as defined)
the color of $u$. The classification vector of the parent $v$ of $u$ includes then
the color information of all its children, including $u$. By our computation of
the classification number of $v$ (described in (C)) we ensure that the roots
of the two trees get the same classification numbers exactly if their respective
colors can be mapped in addition to having isomorphic subtrees.

To maintain the classification numbers and vectors we still use the data schema 
of Lemma~\ref{lem:dataSchema}, but have to account for the additional space
of the colors.

Without modifications the running time increases to $O(n + b)$ due to
the construction of the balanced parenthesis representation of the ordinal tree
with access to its colors and Lemma~\ref{lem:dataSchema}.
For the same reasons our space bound increases to $O(n + b)$ bits.
\end{proof}

\section{Forest Isomorphism}\label{sec:forest-iso}
We now describe how to extend our algorithm to forests.
The idea is to perform a problem reduction, from (colored) forest to colored tree isomorphism by
carefully transforming the given forests into trees without altering their isomorphic properties.
We then solve colored tree isomorphism which is also a solution for forest isomorphism.

A naive approach to transforming a forest into a single tree is to create a new node and 
connect it to an arbitrary node from each tree in the forest. However, this method is unreliable, 
as selecting the incorrect node for connection might lead to generating non-isomorphic trees 
from isomorphic forests, or vise-versa.

To address this, we utilize the fact that the bijective mapping between two isomorphic trees, 
$T$ and $T'$, includes a mapping that aligns their centers. These centers can be either 
{\em unicentric} (a single node) or {\em bicentric} (a pair of nodes). In unicentric cases, 
the mapping between centers is straightforward and unique. However, in bicentric situations, 
the correct mapping is ambiguous due to the two possible ways to align the pair of nodes at each center.

To resolve this ambiguity, we introduce a colored unicentric node~$w$ 
whenever a bicentric pair is identified, as per Rule~1 below 
(see Def.~\ref{def:forest-to-tree-reduction} and Fig.~\ref{fig:forest-to-tree-transformation}).
This step allows us to connect the unicentric node of each tree in the forest to a new, 
single node~$r$, effectively forming one comprehensive tree, as defined in Rule~2.

\begin{definition}[Forest to tree reduction]\label{def:forest-to-tree-reduction}
	Choose a new arbitrary color~$c$.
	\begin{description}
	\item[Rule 1] For each connected component with a bicentric pair of nodes $\{u, v\}$, 
introduce a new $c$-colored node $w$, connecting $w$ to both $u$ and $v$ 
while disconnecting $u$ from $v$.
	\item[Rule 2] Insert a $c$-colored node $r$ and connect it to the center of each connected component.
\end{description}
\end{definition}

\begin{figure}[h!]
	\centering
	\begin{subfigure}{0.5\textwidth}
		\centering
		\begin{tikzpicture}
  \node[circle, draw, minimum size=6pt, inner sep=0] (a) at (0,0) {};
  \node[circle, draw, minimum size=6pt, inner sep=0] (a0) at (-0.3,-0.5) {};
  \node[circle, draw, minimum size=6pt, inner sep=0] (a1) at (0.3,-0.5) {};
  \node[circle, draw, minimum size=6pt, inner sep=0] (a2) at (-0.55,-1) {};
  \node[circle, draw, minimum size=6pt, inner sep=0] (a3) at (-0.1,-1) {};
  \node[circle, draw, minimum size=6pt, inner sep=0] (a4) at (0.3,-1) {};

  \draw (a) -- (a0);  
  \draw (a) -- (a1);
  \draw (a0) -- (a2);
  \draw (a0) -- (a3);
  \draw (a1) -- (a4);
  
  \node[circle, draw, minimum size=6pt, inner sep=0, label=left:{\small$u$}] (b0) at (1.2,-0.5) {};
  \node[circle, draw, minimum size=6pt, inner sep=0, label=right:{\small$v$}] (b1) at (1.8,-0.5) {};
  \node[circle, draw, minimum size=6pt, inner sep=0] (b2) at (1,-1) {};
  \node[circle, draw, minimum size=6pt, inner sep=0] (b3) at (1.4,-1) {};
  \node[circle, draw, minimum size=6pt, inner sep=0] (b4) at (1.8,-1) {};
     
  \draw (b0) -- (b1);
  \draw (b0) -- (b2);
  \draw (b0) -- (b3);
  \draw (b1) -- (b4);
		\end{tikzpicture}
	\end{subfigure}%
	\begin{subfigure}{0.5\textwidth}
		\centering
				\begin{tikzpicture}

  \node[circle, draw, minimum size=6pt, inner sep=0, fill=red, label=right:{\small$r$}] (r) at (0.8,0.5) {};
  \node[circle, draw, minimum size=6pt, inner sep=0, fill=red, label=right:{\small$w$}] (w) at (1.5,0) {};
  	
  \node[circle, draw, minimum size=6pt, inner sep=0] (a) at (0,0) {};
  \node[circle, draw, minimum size=6pt, inner sep=0] (a0) at (-0.3,-0.5) {};
  \node[circle, draw, minimum size=6pt, inner sep=0] (a1) at (0.3,-0.5) {};
  \node[circle, draw, minimum size=6pt, inner sep=0] (a2) at (-0.55,-1) {};
  \node[circle, draw, minimum size=6pt, inner sep=0] (a3) at (-0.1,-1) {};
  \node[circle, draw, minimum size=6pt, inner sep=0] (a4) at (0.3,-1) {};

  \draw (a) -- (a0);  
  \draw (a) -- (a1);
  \draw (a0) -- (a2);
  \draw (a0) -- (a3);
  \draw (a1) -- (a4);

  \node[circle, draw, minimum size=6pt, inner sep=0, label=left:{\small$u$}] (b0) at (1.2,-0.5) {};
  \node[circle, draw, minimum size=6pt, inner sep=0, label=right:{\small$v$}] (b1) at (1.8,-0.5) {};
  \node[circle, draw, minimum size=6pt, inner sep=0] (b2) at (1,-1) {};
  \node[circle, draw, minimum size=6pt, inner sep=0] (b3) at (1.4,-1) {};
  \node[circle, draw, minimum size=6pt, inner sep=0] (b4) at (1.8,-1) {};

  \draw (b0) -- (b2);
  \draw (b0) -- (b3);
  \draw (b1) -- (b4);

  \draw (r) -- (a);
  \draw (r) -- (w);
  \draw (w) -- (b0);
  \draw (w) -- (b1);

		\end{tikzpicture}
	\end{subfigure}
	\caption{Left: a forest where one tree has a center consisting of one node
	and a tree that has a center consisting of two nodes $\{u, v\}$. Right: The
	forest transformed to a unique tree, where centers consisting of two nodes
	are ``extended'' with a new colored center~$w$ and all trees are connected
	to a colored root $r$.}\label{fig:forest-to-tree-transformation}
\end{figure}

We want to remark that all newly introduced nodes have to be colored distinctively
from the already existing nodes, in the case we are working with node-colored
forests. This coloring ensures that these nodes can only be mapped to other
introduced nodes. Consequently, Rule~1 does not alter any tree in a forest in
such a way that it becomes isomorphic to a tree in another forest that already
possessed a unicentric node. We refer to a reduction of a forest $F$ to a tree
$T$ according to Definition~\ref{def:forest-to-tree-reduction} as a \textit{tree
reduction of $F$}.

The next lemma shows that our reduction from forest to tree isomorphism
preserves their isomorphic property.

\begin{lemma}\label{lem:safe-transformation} 
	Two forests $F, F'$ are isomorphic exactly if their respective
	tree reductions $T, T'$ are isomorphic.
\end{lemma}
\begin{proof}
To show our claim we proof both directions.
\begin{itemize}
	\item[$\Rightarrow$]
	To demonstrate this implication, we must establish a bijection $f'$ between $T$ and $T'$.
	Assume that $f$ is an existing bijection between $F$ and $F'$.
	
	Observe that our transformation exclusively adds colored nodes, but does
	not remove any nodes.
	Our goal is to extend $f$ to $f'$ by incorporating mappings of the 
	newly introduced colored nodes.

	First, consider the trees (in the forests $F$ and $F'$) that have a
	unicentric node. Those trees are mapped by $f$ from $F$ to $F'$ and the same
	mapping can be used as mapping $f'$ to map these nodes from $T$ to $T'$.
	Next, let us consider the trees that are bicentric. Recall that $f$ maps two
	center nodes $u$ and $v$ of a bicentric tree of $F$ to two center nodes of a
	bicentric tree $u'$ and $v'$ in $F'$. Then our construction introduces
	colored nodes~$w$ and~$w'$ to $T_F$ and $T_{F'}$, respectively, connected to
	each of these center nodes. First note that when removing all edges between
	nodes that form a bicenter in $F$ and $F'$ during the process of creating a
	tree reduction, the isomorphism $f$ remains a valid isomorphism. Next, our
	tree reduction introduces respective parent nodes $w$ and $w'$ of $\{u, v\}$
	and $\{u', v'\}$. Again, $f$ can be trivially extended to map $w$ to $w'$ as
	$w$ and $w'$ have an equal number of isomorphic subtrees. Finally, we add
	root $r$ and $r'$ to produce $T$ and $T'$, and simply extending $f$ to map
	$r$ to $r'$ produces our desired isomorphism $f'$ that maps $T$ to $T'$.

%

	\item[$\Leftarrow$] 
	To demonstrate this implication, we aim to establish a bijection 
	$f$ between $F$ and $F'$. Assume that $f'$ is a bijection between $T$ and $T'$.
	Then restricting $f'$ to the nodes of $F$ gives us a valid bijection from
	$F$ to $F'$.

	The transformation introduces mappings in $f'$ between colored nodes. 
	Since no nodes are removed in this process, eliminating these specific 
	colored node mappings from $f'$ yields the bijection $f$.
\end{itemize}
Showing both directions concludes our proof.
\end{proof}

We now focus on the realization of the reduction rules, i.e., on an algorithm
 and a data structure representing the tree resulting from the reduction rules
 of Def.~\ref{def:forest-to-tree-reduction}. The technical parts and the
 management of colors are restricted to the proof of
 Lemma~\ref{lem:unique-tree-transformation}, and we focus next on the crucial
 ideas of our approach.
 
Recall that a tree can be represented by a balanced parenthesis
for which, by Lemma~\ref{lem:balanced-tree}, an auxiliary data structure
can be created that allows tree navigation. 
Our goal is to create a balanced parenthesis representation for the tree resulting
from the application of our reduction rules.

Concerning our reduction rules, observe that for each vertex $v$ the parenthesis
representation of each maximal induced subtree rooted at a child
of $v$ is stored in-between the open and the matched closed parenthesis of~$v$.
This means that we can realize Rule~2, where we have to introduce a new root
node $r$, by writing an open and closed parenthesis and printing the parenthesis
representation of each tree in the forest in-between the open and closed
parenthesis of~$r$. Depending on the node with which we start to print the
parenthesis representation of each tree in the forest the node becomes the child
of~$r$. Hence, we determine the center of the tree first. If it is unicentric we
can start printing the balanced representation of the tree starting with the
center. Otherwise, it is bicentric consisting of nodes $\{u, v\}$, and we have
to apply Rule~1. As described in Rule~1 we have to introduce a new node $w$,
which we do by writing an open parenthesis and print the balanced parenthesis
representation of the maximal subtree $u$ and $v$ ignoring the edge $\{u, v\}$,
and conclude with a closed parenthesis for~$w$. In the balanced parenthesis
representation $v$ and $u$ become children of $w$ and have no edge connecting
them, realizing Rule~1.

\begin{lemma}\label{lem:unique-tree-transformation}
	Given an $n$-node colored forest $F$
        with $b$ being the bits used to store all colors,
        there is a linear-time $O(n + b)$-bit algorithm
	computing an $O(n+b)$-bit data structure that gives access to a tree $T$ obtained from the 
	tree reduction of $F$ (Def.~\ref{def:forest-to-tree-reduction}) and allows
	access to the color of each node.
\end{lemma}
\begin{proof}
We represent the structure of the resulting tree $T$ by using a balanced parenthesis representation.
To store the representation we use a bit vector $B_T$, for which we allocate $(3n + 2)$-bits.
The space requirement arises from $2x$ bits to represent $x$ nodes,
another $x$ bits, since we add (by Rule~1 of Def.~\ref{def:forest-to-tree-reduction}) one node per connected 
component with a center that is bicentric
of which we can have at most $\lfloor x/2 \rfloor$ of,
and another $2$ bits since we add one node as the new root.

Recall that during the process of constructing a tree reduction we introduce a
new color $c$ for the nodes we insert. To distinguish $c$ from the colors
already present in $F$, we represent $c$ with a single bit set to $1$, and
prepend $0$ to all colors of $F$. This increases the required bits to store all
colors by $n$-bits. We concretely store these new colors in a
single bit vector $B_C$ with $|B_C|=O(n+b)$ combined static space allocation
(Section~\ref{sec:preone}), which allows us random access queries to the colors
stored in $B_C$.

We now focus on the construction of the balanced parenthesis representation
and assume the data structures above are handled appropriately.
We first create a root $r$ by writing an open 
parenthesis into $B_T$ and write the color $c$~in $B_C$.

Iterate over the forest, and for each tree determine its center (using
Lemma~\ref{lem:select-root}), which is either of the form $C_1 = \{v\}$ or $C_2
= \{v, u\}$. In the case of $C_1$ traverse the respective tree via a DFS to
write its parenthesis representation (rooted at $v$) into $B_T$ and the respective
node colors
into $B_C$. In case $C_2$ first create the artificial
node~$w$ by writing an open parenthesis into $B_T$ and write $c$ into $B_C$ as
the color of~$w$. Then ignoring the edge $\{v, u\}$, which splits the tree into
two subtrees, first print the parenthesis representation of the subtree rooted
at $v$ followed by the parenthesis representation of the subtree rooted at $u$.
Then, close the parenthesis for~$w$. After dealing with all trees, close the
parenthesis for $r$.

$B_T$ and $B_C$ use $O(n)$ and $O(n+b)$ bits, respectively. The remaining space
for variables and pointers is negligible. Thus, constructing the tree reduction
$F$ takes $O(n + b)$ time and bits.
\end{proof}

By Lemma~\ref{lem:unique-tree-transformation} we can solve colored forest isomorphism
as follows. First transform a given forest into its unique colored tree and then solve
colored tree isomorphism. We so can conclude the next theorem.

\begin{theorem}\label{thm:forests} 
    Given two colored $n$-node
	forests $F_1$ and $F_2$ with $b$ being the sum of bits to write down the colors of
        all nodes of $F_1$ and $F_2$,
        there is an algorithm that outputs if $F_1$ and $F_2$
	are isomorphic in linear-time using $O(n + b)$ bits.
\end{theorem}

\section{Outerplanar Graphs}\label{sec:outerplanar}
In this section we address the isomorphism problem for outerplanar graphs.
Outerplanar graphs are planar graphs with an embedding such that all vertices are
incident to one face, usually called the \textit{outer face}.
We start with a simple algorithm for maximal outerplanar graphs, which are
outerplanar graphs to which no edge can be added without invalidating (outer)planarity.
Afterwards, we address biconnected outerplanar graphs, where
biconnected means that the graph remains connected when removing an arbitrary
vertex. Finally, we combine the results for biconnected outerplanar graphs with
our previous results on trees/forests to solve the isomorphism problem
on all outerplanar graphs. We start with some additional preliminaries
and useful data structures. Recall that we use $n$ to denote
the number of vertices of the graph at hand and $m$ the number of edges.

\textbf{Partial Embeddings.} We require that all outerplanar graphs are
\textit{partially embedded}, which we define as being able to distinguish
between edges part of the outer face, and edges part of some inner face.
Concretely, we require the ability to walk around the outer face. If the graph
is given with a geometric/combinatorial embedding, it is trivially partially
embedded. Otherwise, we can use the space-efficient algorithm of Kammer and
Meintrup~\cite{KamM23} to construct and store a partial embedding using $O(n)$
time and bits. Note that the algorithm of Kammer and Meintrup requires the graph
to be given via adjacency arrays with crosspointers. 
Thus, we assume without loss of generality that all
graphs of this section are partially embedded.

\textbf{Block-cut Trees.} A block-cut tree of a graph $G$ is a tree
representation of the biconnected components of $G$. More exactly, it contains a
\textit{cut node} for each cut vertex of $G$ and a \textit{block node} for each
biconnected component, also called a \textit{block}. A cut node is adjacent to a
block node exactly if it is contained in the block the respective block node
represents. Note that each cut node is adjacent to at least two block nodes. Due
to Kammer et al.~\cite{KamKL19} we know there exists an algorithm that finds all
cut vertices of a given graph $G$ and constructs a compact representation of the
block-cut tree. In particular, the data structure they construct allows
iterating over the nodes of the block-cut tree in DFS order and output the block
and cut vertex the respective nodes represent. In addition, each block of the
block-cut tree can be accessed in isolation, e.g., for a vertex $u$ in a block
$B$ we can iterate over all neighbors $N(u) \cap B$ in constant time per
element. Their data structure requires $O(n+m)$ time and bits to construct and
store. Note that the construction of their data structure requires
the input graph to be given via adjacency arrays with crosspointers. We refer
to a data structure that allows these operations as a \textit{compact block-cut
tree}. Internally the data structure effectively stores a compressed DFS tree
together with a special coloring of the vertices and edges. As we work with
outerplanar graphs in this section, the number of edges $m=O(n)$. We summarize
their results in the following lemma.

\begin{lemma}[\cite{KamKL19}]\label{lem:compactbct}
    Let $G$ be a connected graph with $n$ vertices and $m$ edges. We can
    construct and store a compact block-cut tree in $O(n+m)$ time and bits.
\end{lemma}

\textbf{Space-efficient subgraph.} Combining well known techniques in the field
of space-efficient algorithms, we can construct a subgraph data structure $G'$
of a given graph $G$ that allows the following operations, in addition to
standard neighborhood operations: adding/removing vertices and adding/removing
edges under the stipulation that a vertex/edge can only be added to $G'$ if it
is contained in $G$. In simpler terms: $G'$ must always remain a subgraph of
$G$. The data structure can be stored with $O(n+m)$ bits and all operations run
in constant time. It is well known that there exists a constant in-degree
(out-degree) \textit{orientation} $\mathcal{O}(G)$ of any given planar graph
$G$, which is constructed by directing each undirected edge of $G$. It can be
constructed in linear time. Combining this orientation with the previously
described subgraph data structure enables constant time adjacency queries
between any two vertices $u, v \in V(G)$. Such a data structure was previously
described by Kammer and Meintrup in~\cite{KammerM22}. We refer to this combined data
structure as \textit{space-efficient dynamic subgraph}.

\begin{lemma}[\cite{KammerM22}]\label{lem:sedyng} Let $G$ be a planar graph. A space-efficient
    dynamic subgraph $G'$ can be constructed and stored for $G$ in $O(n)$ time and
    bits.
\end{lemma}

\textbf{Duval's algorithm.}
There is a well-known algorithm due to Duval~\cite{Duval83} that takes as input
a string $S=(c_1, c_2, \ldots, c_k)$ and outputs the index $i$ such that
$S'=(c_i, c_{i+1}, \ldots, c_1, c_2, \ldots, c_{i-1})$ is the lexicographically
minimal rotation of $S$. The algorithm runs in $O(k)$ time and uses $O(\log^2
k)$ bits of auxiliary space. The algorithm does not require that
the input string is represented as an array of characters. For this we introduce the definition of so-called virtual
string: We call a deterministic function $s: [k] \rightarrow \Sigma$ a
\textit{virtual string of $S$} exactly if $s(i)$ evaluates to the character
$c_i$ in string $S=(c_1, c_2, \ldots, c_k)$. We require that $s$ can be
evaluated in constant time. All our use-cases involve strings where each
character can be represented using at a constant number of words in the word-RAM
model.

\begin{lemma}{(Duval's algorithm~\cite{Duval83})}\label{lem:duval} 
	Let $S$ be a string of length $k$ represented by
	a virtual string $s$. There is an algorithm that takes as input $s$ and
	computes the index $i$ of the character $c_i$ of $S$ such that $S'=(c_i,
	c_{i+1}, \ldots, c_1, c_2, \ldots, c_{i-1})$ is the lexicographically
	minimal rotation of $S$. The algorithm runs in $O(k)$ time and uses
	$O(\log^2 k)$ bits of auxiliary space
\end{lemma}

\subsection{Warm up: Maximal Outerplanar Graphs}

As mentioned, maximal outerplanar graphs are outerplanar graphs to which no
additional edges can be added without invalidating (outer)planarity. Due to
Beyer and Mitchell~\cite{MitBJ79} an algorithm exists that takes a maximal
outerplanar graph and computes a string representation, such that for any two
maximal outerplanar graphs $G$ and $G'$, the string representation is the same
exactly if $G$ and $G'$ are isomorphic. This is referred to as a
\textit{canonical representation}. The algorithm makes use of the fact that
every biconnected outerplanar graphs has a unique hamilton cycle (all maximal
outerplanar graphs are biconnected), where a \textit{hamilton cycle} is a simple
cycle in the graph that contains all vertices. In biconnected outerplanar graphs
this is simply a walk around the outer face in clockwise or counterclockwise
order. For a maximal outerplanar graph $G$ Beyer and Mitchell's algorithm first
computes a hamilton cycle $H$. Then a so-called \textit{degree sequence $S$} of
$H$ is computed, such that $S$ is a representation of $H$ that only stores the
degree of each
vertex instead of the vertex itself.

Beyer and Mitchell showed that $S$ is canonical when taking into account
possible rotation and reversal of the hamilton cycle. To test isomorphism
between two maximal outerplanar graphs $G$ and $G'$, they describe an algorithm
that first computes a degree sequence $S$ and $S'$ of $G$ and $G'$,
respectively. Then, the sequences are checked for equality, such that two
sequences are defined as equal if one of them can be rotated and/or reversed to
be equal to the other. We replace the final step with an arguably simpler
procedure (which effectively does the same). Among all possible rotations and
reversals of a degree sequence $S$ we denote with $S_{\texttt{min}}$ the
lexicographically minimal such degree sequence. The lexicographically minimal
degree sequence gives us a canonical representation. Isomorphism can then be
solved by simply comparing lexicographically minimal degree sequences directly. 

In the following let $G$ be a maximal outerplanar graph. Simply storing degree
sequences requires $\Theta(n \log n)$ bits, as we must store the degree of $n$
vertices. To reduce the space to $O(n)$, we store each value as a
self-delimiting number. As the sum of all degrees in a maximal outerplanar graph
is $O(n)$, this requires $O(n)$ bits. To compute the lexicographically minimal
degree sequence $S_{\texttt{min}}$ we first compute two sequences
$S_{\texttt{cw}}$ and $S_{\texttt{cc}}$, via a walk on the outer face in
clockwise and counterclockwise direction, respectively. Then, we compute the
respective lexicographically minimal rotation $S'_{\texttt{cw}}$ and
$S'_{\texttt{cc}}$ using Duval's algorithm (Lemma~\ref{lem:duval}). Then
$S_{\texttt{min}}=\texttt{min}\{S'_{\texttt{cw}}, S'_{\texttt{cc}}\}$. The
respective initial degree sequences are computed by a simple walk around the
outer face. For this we need to know \textit{outer edges}, i.e., the edges
incident to the outer face. The sketched algorithm requires a partial embedding
of $G$. As mentioned, one can compute such an embedding
space-efficiently~\cite{KamM23}.

We want to remark that for maximal outerplanar graphs we can show a much simpler
space-efficient partial embedding algorithm, presented as the next lemma.

\begin{lemma}\label{lem:mopemb}
    Let $G=(V, E)$ be a maximal outerplanar graph. We can compute and store a
    partial embedding of $G$ in $O(n)$ time and bits.
\end{lemma}
\begin{proof}
    In the following we make use of the well-known fact that maximal outerplanar
    graphs contain at least $2$ vertices of degree $2$.

    We require the following variables. First, a choice dictionary
    (Lemma~\ref{lem:cd}) $Q$ that acts as a queue of vertices of degree $2$,
    initially containing all such vertices of $G$. Secondly, two choice
    dictionaries $F_{\text{inner}}$ and $F_{\text{outer}}$ used for marking
    edges as either "inner" or "outer", initially these two choice dictionaries
    are empty. And thirdly, a dynamic space-efficient subgraph data
    structure~(Lemma~\ref{lem:sedyng}) storing a subgraph $G'$ of $G$. Initially
    $G'=G$. We then iteratively pop vertices $u$ from $Q$ (until $Q=\emptyset$)
    and do the following procedure: let $v, w$ be the neighbors of $u$ in $G'$.
    Add the edge $\{u, v\}$ to $F_{\text{outer}}$ if it is not contained in
    $F_{\text{inner}}$, and analogously for edge $\{u, w\}$. Then, add the edge
    $\{w, v\}$ to $F_{\text{inner}}$. Note that this edge exists as all inner
    faces of a maximal outerplanar graph contain exactly $3$ vertices. Remove
    the edges $\{u, v\}$ and $\{u, w\}$ from $G'$. Update $Q$ by adding $u$ and
    $w$ exactly if their respective degrees equals $2$ in $G'$. Note that after
    each iteration the graph $G'$ remains maximal outerplanar (we view isolated
    vertices in $G'$ as deleted). This means we always find a vertex of degree $2$ in
    the next iteration (or we are finished).

    Once this procedure terminates, $F_{\text{outer}}$ contains all edges part
    of the outer face, and $F_{\text{inner}}$ contains all other edges. Each
    iteration of the aforementioned loop requires $O(1)$ time and there are
    $O(n)$ iterations total. All data structures require $O(n)$ bits.
\end{proof}

Using the previous lemma we can easily implement the outlined algorithm to
compute a canonical representation $G_{\texttt{can}}=S_{\texttt{min}}$ of a
maximal outerplanar graph in linear time and bits.

\begin{theorem}\label{thm:maxoutercan}
    Let $G$ be a maximal outerplanar graph. We can compute and store a canonical
    representation $G_{\texttt{can}}$ of $G$ in $O(n)$ time and bits.
\end{theorem}

A simple isomorphism algorithm follows by comparing the canonical
representations.

\begin{corollary}\label{cor:maxouteriso}
    Let $G$ and $G'$ be two maximal outerplanar graphs. We can test isomorphism
    of $G$ and $G'$ in $O(n)$ time and bits.
\end{corollary}
\begin{proof}
    Using Theorem~\ref{thm:maxoutercan}, compute canonical representations
    $G_{\texttt{can}}$ and $G'_{\texttt{can}}$ of $G$ and $G'$, respectively.
    Then $G$ and $G'$ are isomorphic exactly if their canonical representations
    are equal.
\end{proof}

\subsection{Biconnected Outerplanar Graphs}

We continue with biconnected outerplanar graphs. Recall that a biconnected graph
is a connected graph where the removal of any single vertex does
not disconnect the graph. On a high level, we handle
biconnected outerplanar graphs similar as maximal outerplanar graphs: we compute
a canonical representation based on a hamilton cycle. Contrary to maximal
outerplanar graphs, where the simple degree sequence uniquely identifies
isomorphic graphs, this is not the case for biconnected outerplanar graphs.
Colbourn and Booth~\cite{ColbournBK81} showed the following canonical
representation of biconnected outerplanar graphs, which we show how to construct
and store space-efficiently. Let $G$ be a biconnected outerplanar graph. First,
construct a hamilton cycle $H$ of $G$. For two vertices $u, v \in V(G)$ denote
with \textit{hamilton distance $h(u, v)$} the length of a shortest path between
$u$ and $v$ that uses only outer edges. To construct a canonical representation
$G_{\texttt{can}}$ of $G$ replace each vertex $u$ in $H$ with the sequence
$S_u=(h(u, v) : v \in N(u))$, ordered in increasing values. We refer to this as
a \textit{hamilton-distance sequence}. Colbourn and Booth showed that such a
hamilton-distance sequence is canonical, up to the choice of the start vertex
and clockwise or counterclockwise direction of the hamilton cycle. Analogous to
the technique used for constructing a canonical representation for maximal
outerplanar graphs, simply choosing the lexicographically smallest hamilton
distance sequence produces a canonical representation $G_{\texttt{can}}$ of $G$.
This again can easily be achieved using Duval's algorithm
(Lemma~\ref{lem:duval}).

Using a standard representation the sequence requires $\Theta(n \log n)$ bits to
be stored. Contrary to the canonical representation for maximal outerplanar
graphs we can not use self-delimiting number to store each number in the
hamilton distance sequence, as the sum of distances is $\omega(n)$ in some
biconnected outerplanar graphs. Instead, we use a technique based on
balanced parentheses to encode the distances, which we refer to as
\textit{compact hamilton-distance sequence}. Intuitively, each hamilton distance
of a vertex $u$ to a vertex $v$ that is contained in the hamilton-distance
sequence is encoded compactly as a pair of matching parenthesis in a balanced
parenthesis string such that their distance in the string is equal (omitting
negligible additional calculations) to the hamilton distance of $u$ and $v$. A
compact hamilton-distance sequence $P$ that represents a hamilton-distance
sequence $H$ allows us to construct a virtual string $h$ of $H$. We then use $h$
to compute the start index of a lexicographically minimal rotation $H'$ of $H$.
Additional details are delegated to the proof of the following lemma. As we
later require to also encode colored biconnected outerplanar graphs, we
trivially extend Colbourn and Booth's algorithm to allow this.

\begin{lemma}\label{lem:bopcan} Let $G$ be a vertex-colored biconnected
    outerplanar graph such that all colors require $b$ bits total to be stored.
    We can construct and store a canonical representation $G_{\texttt{can}}$ of
    $G$ in $O(n+b)$ time and bits.
\end{lemma}
\begin{proof}
    We start by describing an algorithm to construct a balanced parenthesis
    string $P$ over the universe $\{[, ], (, )\}$ such that $P$ is a
    compact hamilton- distance sequence of $G$. The sequence $P$ is built
    incrementally. Initially $P=\emptyset$. In the following we denote with
    $N^*(u)$ the set of neighbors of a vertex $u \in V$ that are connected to
    $u$ via an inner edge. We first mark each vertex of $V$ as unvisited. To
    construct $P$ we walk around the outer face of $G$ and for each vertex $u
    \in V$ we first add a '$[$' to $P$. Then, for each $v \in N^*(u)$, we do the
    following: if $v$ was already visited during our aforementioned walk, add a
    '$)$' to $P$, otherwise add a '$($' to $P$. Once all neighbors are processed
    this way, add '$]$' to $P$ and mark $u$ as visited. Do this procedure until
    all vertices of $V(G)$ are visited. 

    The runtime of the algorithm is linear in the number of edges of $G$ which
    is $O(n)$. The string $P$ contains $O(n)$ elements over the universe $\{[,
    ], (, )\}$ and thus uses $O(n)$ bits. To mark vertices as visited or
    unvisited a simple bit vector suffices. To handle the colors simple store
    the color of each vertex $u \in V(G)$ in the sequence $P$ directly after
    processing $N^*(u)$. This uses an additional $O(b)$ time and bits.

	Using the data structure of Lemma~\ref{lem:balanced-tree} we can evaluate
	the distance between matching pairs of parenthesis in constant time.
	Recall that the data structure can easily be extended to allow
	different types of parenthesis (Section~\ref{sec:prelim}). This directly
	gives a virtual string representation $h$ of a hamilton distance sequence
	$H$. We then use $h$ to compute the start vertex of a lexicographically
	minimal hamilton distance sequence $H'$ of $H$. Let $u$ be the computed
	start vertex. We then run the algorithm described in the first paragraph
	again, this time starting at vertex $u$. The balanced parenthesis string
	$P'$ computed from this then represents $H'$.

	We must do this entire process twice: once based on a walk around the outer
	face in clockwise and once based on a walk in counterclockwise direction.
	Assume that $P_{\texttt{cw}}$ and $P_{\texttt{cc}}$ are the respectively
	balanced parenthesis strings for each direction. We then choose
	$G_{\texttt{can}}=\texttt{min}\{P_{\texttt{cw}}, P_{\texttt{cc}}\}$. In
	total, this requires $O(n+b)$ time and bits.
\end{proof}

\subsection{Outerplanar Graphs}

We now generalize the previous section from biconnected outerplanar graphs to
connected outerplanar graphs. We first sketch the idea of an isomorphism
algorithm for outerplanar graphs. Roughly speaking, we reduce the problem of
testing the isomorphism of two connected outerplanar graphs $G$ and $G'$ to
testing the isomorphism of the respective block-cut trees $T$ and $T'$. An
algorithm due to Booth and Colbourn~\cite{ColbournBK81} first computes such
block-cut trees, and then color each node that represents a block $B$ with the
hamilton distance sequence of $B$. One can then simply use a colored tree
isomorphism algorithm on the respective colored block-cut trees. Our
space-efficient algorithm works almost analogously: we compute a compact
block-cut tree representation $T$ and $T'$ of $G$ and $G'$, respectively. We
then construct respective colored block-cut trees, stored using a parenthesis
representation together with the canonical representation of
Lemma~\ref{lem:bopcan} as the respective colors assigned to nodes of the tree.
Afterwards, we can use the colored tree isomorphism algorithm of
Theorem~\ref{lm:colored-trees}. In the case that the two given outerplanar
graphs are not connected, we instead use the algorithm 
of Theorem~\ref{thm:forests} for colored forests on
the forest of block-cut trees. Missing details are delegated to the proof.

\begin{theorem}\label{thm:conopiso}
    Let $G$ and $G'$ be two partially embedded outerplanar graphs. We
    can answer if $G$ is isomorphic to $G'$ using $O(n)$ time and bits.
\end{theorem}
\begin{proof}
    We describe the algorithm in three phases. The first two phases are only
    described for the graph $G$. It works analogously for $G'$. We assume that
    neither $G$ nor $G'$ is biconnected, as otherwise we can directly use
    Lemma~\ref{lem:bopcan}.

    \textbf{Finding a canonical root block $B_r$.} We first require to find a
    canonical root block of the block-cut tree. For this we use the following
    simple algorithm. Construct a compact block-cut tree $T$ of $G$ using
    Lemma~\ref{lem:compactbct}. Using Lemma~\ref{lem:select-root} find the
    center of $T$. If the center consists of two nodes, one node must be a block
    node representing some block $B$. This is due to each block node being
    adjacent to only cut node, and vice-versa. Then $B$ is our
    root-block $B_r$. If the center consists of one node $w$, then $w$ might
    either be a cut node or block node. If $w$ is a cut node, we
    introduce a single dummy vertex $u_{\texttt{dummy}}$ to $V(G)$ that is
    adjacent to only the cut vertex $u_w$ represented by node $w$. We then use
    the block $B_r=\{u_{\texttt{dummy}}, u_w\}$ of $G$ as our root block. Note
    that the dummy vertex $u_w$ can be trivially stored using $O(\log n)$ bits,
    together with the newly introduced block $B_r$. In the case that the node
    $w$ is a block node, we simply use the block mapped to $w$ as our root-block
    $B_r$.

    \textbf{Constructing a parenthesis representation.} Assume we have found
    the root block $B_r$ of $G$ as described in the previous paragraph.
    Starting at an arbitrary vertex $u \in B_r$ construct a compact block-cut
    tree $T$ of $G$. Due to our choice of start vertex this block-cut tree is now
    rooted at $B_r$. We iteratively build a parenthesis representation $P$.
    Initially, $P$ is an empty bit vector which we view as a stack. 
    Traverse $T$ in recursive DFS-order and when encountering a node $w$ of $T$ do the
    following:
    \begin{itemize}
        \item $w$ is block node: Let $B$ be the block represented by $w$ and let
            $u \in B$ be the cut vertex represented by the parent node $w'$ of
            $w$ (if it exists). Color the vertices of $B\setminus\{u\}$ that are
            cut vertices in $G$ as \texttt{cut} and the other vertices of
            $B\setminus\{u\}$ as \texttt{normal}. If the vertex $u$ exists,
            color it as \texttt{parent-cut}. Use Lemma~\ref{lem:bopcan} to
            construct a (colored) canonical representation $B_{\texttt{can}}$ of
            $B$. Push an opening parenthesis on $P$ followed by a bit
            indicating that $w$ is a block node and then followed by
            $B_{\texttt{can}}$ and recursively traverse each child node $w''$ of
            $w$. Note that each child node $w''$ of $w$ is a cut node
            representing a cut vertex $v \in B$. We continue the recursion such
            that the ordering of the cut vertices $v$ around the hamilton cycle
            of $B_{\texttt{can}}$ is respected. When returning from the
            recursion, add a closing parenthesis to $P$.
        \item $w$ is a cut vertex node: We add an opening parenthesis to $P$
            followed by a bit indicating that $w$ is a cut vertex.
            Recursively continue for all child nodes $w'$ of $w$ in an arbitrary order.
            When returning from the recursion, add a closing parenthesis to $P$.
    \end{itemize}

    \textbf{Testing for isomorphism.} Assume parenthesis representations $P$
    and $P'$ are constructed for $G$ and $G'$, respectively, as outlined in the
    previous paragraph. We can now directly use Theorem~\ref{lm:colored-trees}
    to test for isomorphism, with a small change: nodes of $P$ (and $P'$) are
    partially ordered, and this ordering must be respected. Nodes
    $w'$ that are children of a block node $w$ have their ordering fixed
	during the construction. To
    accommodate this one must simply skip the sorting the children of $w$ during
    the algorithm of Theorem~\ref{lm:colored-trees}.

    Each phase runs in linear time and uses $O(n)$ bits total. The correctness
    was shown in~\cite{ColbournBK81}.
\end{proof}

\bibliographystyle{elsarticle-num}
\bibliography{main}

\begin{thebibliography}{10}
\expandafter\ifx\csname url\endcsname\relax
  \def\url#1{\texttt{#1}}\fi
\expandafter\ifx\csname urlprefix\endcsname\relax\def\urlprefix{URL }\fi
\expandafter\ifx\csname href\endcsname\relax
  \def\href#1#2{#2} \def\path#1{#1}\fi

\bibitem{KamMS23}
F.~Kammer, J.~Meintrup, A.~Sajenko, Sorting and ranking of self-delimiting numbers with applications to tree isomorphism, in: Proc. of the 34th International Workshop on Combinatorial Algorithms ({IWOCA} 2023), Vol. 13889 of Lecture Notes in Computer Science, Springer, 2023, pp. 356--367.
\newblock \href {https://doi.org/10.1007/978-3-031-34347-6\_30} {\path{doi:10.1007/978-3-031-34347-6\_30}}.

\bibitem{IrnB04}
C.~Irniger, H.~Bunke, Decision tree structures for graph database filtering, in: Structural, Syntactic, and Statistical Pattern Recognition, Springer, 2004, pp. 66--75.
\newblock \href {https://doi.org/10.1007/978-3-540-27868-\_6} {\path{doi:10.1007/978-3-540-27868-\_6}}.

\bibitem{BaiC75}
H.~S. Baird, Y.~E. Cho, An artwork design verification system, in: Proc. 12th conference on Design {(DAC 1975)}, IEEE, 1975, p. 414–420.

\bibitem{AsanoEK13}
T.~Asano, A.~Elmasry, J.~Katajainen, Priority queues and sorting for read-only data, in: Proc. 10th International Conference on Theory and Applications of Models of Computation ({TAMC} 2013), Vol. 7876 of Lecture Notes in Computer Science, Springer, 2013, pp. 32--41.
\newblock \href {https://doi.org/10.1007/978-3-642-38236-9\_4} {\path{doi:10.1007/978-3-642-38236-9\_4}}.

\bibitem{BorC82}
A.~Borodin, S.~A. Cook, A time-space tradeoff for sorting on a general sequential model of computation, {SIAM} J. Comput. 11~(2) (1982) 287--297.
\newblock \href {https://doi.org/10.1137/0211022} {\path{doi:10.1137/0211022}}.

\bibitem{Chan10}
T.~M. Chan, Comparison-based time-space lower bounds for selection, {ACM} Trans. Algorithms 6~(2) (2010) 26:1--26:16.
\newblock \href {https://doi.org/10.1145/1721837.1721842} {\path{doi:10.1145/1721837.1721842}}.

\bibitem{Fre87}
G.~N. Frederickson, Upper bounds for time-space trade-offs in sorting and selection, J. Comput. Syst. Sci. 34~(1) (1987) 19--26.
\newblock \href {https://doi.org/10.1016/0022-0000(87)90002-X} {\path{doi:10.1016/0022-0000(87)90002-X}}.

\bibitem{Hagerup98}
T.~Hagerup, Sorting and searching on the word {RAM}, in: Proc. 15th Annual Symposium on Theoretical Aspects of Computer Science ({STACS} 98), Vol. 1373 of Lecture Notes in Computer Science, Springer, 1998, pp. 366--398.
\newblock \href {https://doi.org/10.1007/BFb0028575} {\path{doi:10.1007/BFb0028575}}.

\bibitem{IsaacS56}
E.~J. Isaac, R.~C. Singleton, Sorting by address calculation, J. {ACM} 3~(3) (1956) 169--174.
\newblock \href {https://doi.org/10.1145/320831.320834} {\path{doi:10.1145/320831.320834}}.

\bibitem{Knuth98a}
D.~E. Knuth, The art of computer programming, , Volume III, 2nd Edition, Addison-Wesley, 1998.

\bibitem{PagR98}
J.~Pagter, T.~Rauhe, Optimal time-space trade-offs for sorting, in: Proc. 39th Annual Symposium on Foundations of Computer Science {(FOCS 1998)}, {IEEE} Computer Society, 1998, pp. 264--268.
\newblock \href {https://doi.org/10.1109/SFCS.1998.743455} {\path{doi:10.1109/SFCS.1998.743455}}.

\bibitem{Bea91}
P.~Beame, A general sequential time-space tradeoff for finding unique elements, {SIAM} J. Comput. 20~(2) (1991) 270--277.

\bibitem{PagP02}
R.~Pagh, J.~Pagter, Optimal time-space trade-offs for non-comparison-based sorting, in: Proc. 13th Annual {ACM-SIAM} Symposium on Discrete Algorithms {(SODA 2002)}, {ACM/SIAM}, 2002, pp. 9--18.

\bibitem{Han18}
Y.~Han, Sorting real numbers in {O(n{\(\surd\)}(log n))} time and linear space, Algorithmica 82~(4) (2020) 966--978.
\newblock \href {https://doi.org/10.1007/s00453-019-00626-0} {\path{doi:10.1007/s00453-019-00626-0}}.

\bibitem{RamRR02}
R.~Raman, V.~Raman, S.~S. Rao, Succinct indexable dictionaries with applications to encoding k-ary trees and multisets, in: Proc. 13th Annual {ACM-SIAM} Symposium on Discrete Algorithms {(SODA 2002)}, {ACM/SIAM}, 2002, pp. 233--242.

\bibitem{AsaIKKOOSTU14}
T.~Asano, T.~Izumi, M.~Kiyomi, M.~Konagaya, H.~Ono, Y.~Otachi, P.~Schweitzer, J.~Tarui, R.~Uehara, Depth-first search using {O(n)} bits, in: Proc. 25th International Symposium on Algorithms and Computation {(ISAAC 2014)}, Vol. 8889 of {LNCS}, Springer, 2014, pp. 553--564.
\newblock \href {https://doi.org/10.1007/978-3-319-13075-0\_44} {\path{doi:10.1007/978-3-319-13075-0\_44}}.

\bibitem{BanerjeeNCSRVSR18}
N.~Banerjee, S.~Chakraborty, V.~Raman, S.~R. Satti, Space efficient linear time algorithms for bfs, dfs and applications, Theor. Comp. Sys. 62~(8) (2018) 1736–1762.
\newblock \href {https://doi.org/10.1007/s00224-017-9841-2} {\path{doi:10.1007/s00224-017-9841-2}}.

\bibitem{ElmHK15}
A.~Elmasry, T.~Hagerup, F.~Kammer, Space-efficient basic graph algorithms, in: Proc. 32nd International Symposium on Theoretical Aspects of Computer Science {(STACS 2015)}, Vol.~30 of LIPIcs, Schloss Dagstuhl -- Leibniz-Zentrum f\"ur Informatik, 2015, pp. 288--301.
\newblock \href {https://doi.org/10.4230/LIPIcs.STACS.2015.288} {\path{doi:10.4230/LIPIcs.STACS.2015.288}}.

\bibitem{Hag18}
T.~Hagerup, Space-efficient {DFS} and applications to connectivity problems: Simpler, leaner, faster, Algorithmica 82~(4) (2020) 1033--1056.
\newblock \href {https://doi.org/10.1007/s00453-019-00629-x} {\path{doi:10.1007/s00453-019-00629-x}}.

\bibitem{BarAHM12}
J.~Barbay, L.~C. Aleardi, M.~He, J.~I. Munro, Succinct representation of labeled graphs, Algorithmica 62~(1) (2012) 224--257.
\newblock \href {https://doi.org/10.1007/s00453-010-9452-7} {\path{doi:10.1007/s00453-010-9452-7}}.

\bibitem{HagKL19}
T.~Hagerup, F.~Kammer, M.~Laudahn, Space-efficient euler partition and bipartite edge coloring, Theor. Comput. Sci. 754 (2019) 16--34.
\newblock \href {https://doi.org/10.1016/j.tcs.2018.01.008} {\path{doi:10.1016/j.tcs.2018.01.008}}.

\bibitem{KamM23}
F.~Kammer, J.~Meintrup, Succinct planar encoding with minor operations, in: Proc. 34th International Symposium on Algorithms and Computation, ({ISAAC} 2023), Vol. 283 of LIPIcs, Schloss Dagstuhl - Leibniz-Zentrum f{\"{u}}r Informatik, 2023, pp. 44:1--44:18.
\newblock \href {https://doi.org/10.4230/LIPICS.ISAAC.2023.44} {\path{doi:10.4230/LIPICS.ISAAC.2023.44}}.

\bibitem{ChoGS18}
J.~Choudhari, M.~Gupta, S.~Sharma, Nearly optimal space efficient algorithm for depth first search, CoRR abs/1810.07259 (2018).
\newblock \href {http://arxiv.org/abs/1810.07259} {\path{arXiv:1810.07259}}.

\bibitem{DattaBK16}
S.~Datta, R.~Kulkarni, A.~Mukherjee, Space-efficient approximation scheme for maximum matching in sparse graphs, in: Proc. 41st International Symposium on Mathematical Foundations of Computer Science {(MFCS 2016)}, Vol.~58 of LIPIcs, Schloss Dagstuhl -- Leibniz-Zentrum f\"ur Informatik, 2016, pp. 28:1--28:12.
\newblock \href {https://doi.org/10.4230/LIPIcs.MFCS.2016.28} {\path{doi:10.4230/LIPIcs.MFCS.2016.28}}.

\bibitem{IzumiO20}
T.~Izumi, Y.~Otachi, {Sublinear-Space Lexicographic Depth-First Search for Bounded Treewidth Graphs and Planar Graphs}, in: 47th International Colloquium on Automata, Languages, and Programming (ICALP 2020), Vol. 168, 2020, pp. 67:1--67:17.
\newblock \href {https://doi.org/10.4230/LIPIcs.ICALP.2020.67} {\path{doi:10.4230/LIPIcs.ICALP.2020.67}}.

\bibitem{KammerM22}
F.~Kammer, J.~Meintrup, Space-efficient graph coarsening with applications to succinct planar encodings, in: 33rd International Symposium on Algorithms and Computation, {ISAAC} 2022, Vol. 248, 2022, pp. 62:1--62:15.
\newblock \href {https://doi.org/10.4230/LIPIcs.ISAAC.2022.62} {\path{doi:10.4230/LIPIcs.ISAAC.2022.62}}.

\bibitem{KammerMS22}
F.~Kammer, J.~Meintrup, A.~Sajenko, Space-efficient vertex separators for treewidth, Algorithmica 84~(9) (2022) 2414--2461.
\newblock \href {https://doi.org/10.1007/s00453-022-00967-3} {\path{doi:10.1007/s00453-022-00967-3}}.

\bibitem{HeegerHKNRS21}
K.~Heeger, A.-S. Himmel, F.~Kammer, R.~Niedermeier, M.~Renken, A.~Sajenko, Multistage graph problems on a global budget, Theoretical Computer Science 868 (2021) 46--64.
\newblock \href {https://doi.org/https://doi.org/10.1016/j.tcs.2021.04.002} {\path{doi:https://doi.org/10.1016/j.tcs.2021.04.002}}.

\bibitem{KamS18g}
F.~Kammer, A.~Sajenko, Space efficient (graph) algorithms, \url{https://github.com/thm-mni-ii/sea} (2018).

\bibitem{AhoHU74}
A.~V. Aho, J.~E. Hopcroft, J.~D. Ullman, The Design and Analysis of Computer Algorithms, Addison-Wesley, 1974.

\bibitem{Lindell92}
S.~Lindell, A logspace algorithm for tree canonization (extended abstract), in: Proc. 24th Annual {ACM} Symposium on Theory of Computing ({STOC 1992}), {ACM}, 1992, pp. 400--404.
\newblock \href {https://doi.org/10.1145/129712.129750} {\path{doi:10.1145/129712.129750}}.

\bibitem{ColbournBK81}
C.~J. Colbourn, K.~S. Booth, Linear time automorphism algorithms for trees, interval graphs, and planar graphs, SIAM Journal on Computing 10~(1) (1981) 203--225.
\newblock \href {https://doi.org/10.1137/0210015} {\path{doi:10.1137/0210015}}.

\bibitem{MitBJ79}
T.~Beyer, W.~Jones, S.~L. Mitchell, Linear algorithms for isomorphism of maximal outerplanar graphs, J. {ACM} 26~(4) (1979) 603--610.
\newblock \href {https://doi.org/10.1145/322154.322155} {\path{doi:10.1145/322154.322155}}.

\bibitem{BauH17}
T.~Baumann, T.~Hagerup, Rank-select indices without tears, in: Proc. 16th International Symposium on Algorithms and Data Structures ({WADS} 2019), Vol. 11646 of Lecture Notes in Computer Science, Springer, 2019, pp. 85--98.
\newblock \href {https://doi.org/10.1007/978-3-030-24766-9\_7} {\path{doi:10.1007/978-3-030-24766-9\_7}}.

\bibitem{KamKL19}
F.~Kammer, D.~Kratsch, M.~Laudahn, Space-efficient biconnected components and recognition of outerplanar graphs, Algorithmica 81~(3) (2019) 1180--1204.
\newblock \href {https://doi.org/10.1007/s00453-018-0464-z} {\path{doi:10.1007/s00453-018-0464-z}}.

\bibitem{Hag19}
T.~Hagerup, {A Constant-Time Colored Choice Dictionary with Almost Robust Iteration}, in: Proc. 44th International Symposium on Mathematical Foundations of Computer Science (MFCS 2019), Vol. 138 of {LIPIcs}, Schloss Dagstuhl -- Leibniz-Zentrum f\"ur Informatik, Dagstuhl, Germany, 2019, pp. 64:1--64:14.
\newblock \href {https://doi.org/10.4230/LIPIcs.MFCS.2019.64} {\path{doi:10.4230/LIPIcs.MFCS.2019.64}}.

\bibitem{KamS18c}
F.~Kammer, A.~Sajenko, Simple $2^f$-color choice dictionaries, in: Proc. 29th International Symposium on Algorithms and Computation ({ISAAC} 2018), Vol. 123 of LIPIcs, Schloss Dagstuhl -- Leibniz-Zentrum f\"ur Informatik, 2018, pp. 66:1--66:12.
\newblock \href {https://doi.org/10.4230/LIPIcs.ISAAC.2018.66} {\path{doi:10.4230/LIPIcs.ISAAC.2018.66}}.

\bibitem{Elias74}
P.~Elias, Efficient storage and retrieval by content and address of static files, J. {ACM} 21~(2) (1974) 246--260.
\newblock \href {https://doi.org/10.1145/321812.321820} {\path{doi:10.1145/321812.321820}}.

\bibitem{Jac89}
G.~Jacobson, Space-efficient static trees and graphs, in: Proc. 30th Annual Symposium on Foundations of Computer Science {(FOCS 1989)}, {IEEE} Computer Society, 1989, pp. 549--554.
\newblock \href {https://doi.org/10.1109/SFCS.1989.63533} {\path{doi:10.1109/SFCS.1989.63533}}.

\bibitem{DBLP:journals/tcs/KatohG22}
T.~Katoh, K.~Goto, In-place initializable arrays, Theor. Comput. Sci. 916 (2022) 62--69.
\newblock \href {https://doi.org/10.1016/j.tcs.2022.03.004} {\path{doi:10.1016/j.tcs.2022.03.004}}.

\bibitem{CorLRS09}
T.~H. Cormen, C.~E. Leiserson, R.~L. Rivest, C.~Stein, Introduction to Algorithms, 3rd Edition, {MIT} Press, 2009.

\bibitem{MunR97}
J.~I. Munro, V.~Raman, Succinct representation of balanced parentheses, static trees and planar graphs, in: Proc. 38th Annual Symposium on Foundations of Computer Science, {(FOCS 1997)}, {IEEE} Computer Society, 1997, pp. 118--126.
\newblock \href {https://doi.org/10.1109/SFCS.1997.646100} {\path{doi:10.1109/SFCS.1997.646100}}.

\bibitem{Navarro16}
G.~Navarro, Compact Data Structures: A Practical Approach, Cambridge University Press, 2016.
\newblock \href {https://doi.org/10.1017/CBO9781316588284} {\path{doi:10.1017/CBO9781316588284}}.

\bibitem{Buss97}
S.~Buss, Alogtime algorithms for tree isomorphism, comparison, and canonization, in: Proc. 5th Kurt Godel Colloquium on Computational Logic and Proof Theory {(KGC 1997)}, Vol. 1289 of {LNCS}, Springer, 1997, pp. 18--33.

\bibitem{he_et_al:LIPIcs:2020:13369}
M.~He, J.~I. Munro, Y.~Nekrich, S.~Wild, K.~Wu, {Distance Oracles for Interval Graphs via Breadth-First Rank/Select in Succinct Trees}, in: 31st International Symposium on Algorithms and Computation (ISAAC 2020), Vol. 181 of Leibniz International Proceedings in Informatics (LIPIcs), Schloss Dagstuhl--Leibniz-Zentrum f{\"u}r Informatik, Dagstuhl, Germany, 2020, pp. 25:1--25:18.
\newblock \href {https://doi.org/10.4230/LIPIcs.ISAAC.2020.25} {\path{doi:10.4230/LIPIcs.ISAAC.2020.25}}.

\bibitem{DBLP:journals/algorithmica/FarzanM14}
A.~Farzan, J.~I. Munro, A uniform paradigm to succinctly encode various families of trees, Algorithmica 68~(1) (2014) 16--40.
\newblock \href {https://doi.org/10.1007/s00453-012-9664-0} {\path{doi:10.1007/s00453-012-9664-0}}.

\bibitem{Harary91}
F.~Harary, Graph theory, Addison-Wesley, 1991.

\bibitem{Duval83}
J.~{Pierre Duval}, Factorizing words over an ordered alphabet, Journal of Algorithms 4~(4) (1983) 363--381.
\newblock \href {https://doi.org/10.1016/0196-6774(83)90017-2} {\path{doi:10.1016/0196-6774(83)90017-2}}.

\end{thebibliography}

\end{document}